\documentclass[12pt,a4paper]{article}
\usepackage[utf8]{inputenc}

\title{%
  Testing the constancy of Spearman's rho in multivariate time series
}
\author{%
  Ivan Kojadinovic  \\
  \small{Laboratoire de math\'ematiques et applications, UMR CNRS 5142} \\
  \small{Universit\'e de Pau et des Pays de l'Adour} \\
  \small{B.P. 1155, 64013 Pau Cedex, France} \\
  \small{\texttt{ivan.kojadinovic@univ-pau.fr}} \\
  \and
  Jean-Fran\c{c}ois Quessy \\
  \small{D\'epartement de math\'ematiques et d’informatique} \\
  \small{Universit\'e du Qu\'ebec \`a Trois-Rivi\`eres} \\
  \small{Trois-Rivi\`eres, Québec, C.P. 500, G9A 5H7 Canada} \\
  \small{\texttt{jean-francois.quessy@uqtr.ca}}
  \and
  Tom Rohmer \\
  \small{Laboratoire de math\'ematiques Jean Leray} \\
  \small{Universit\'e de Nantes} \\ 
  \small{B.P. 92208, 44322 Nantes Cedex 3, France} \\
  \small{\texttt{tom.rohmer@univ-nantes.fr}} \\[5mm]
}

\usepackage{fullpage}
\usepackage{amssymb}
\usepackage{amsmath}
\usepackage{amsthm}
\usepackage[round]{natbib}
\usepackage{url}
\usepackage{graphicx}
\usepackage{multirow}
\usepackage{setspace}
\usepackage{enumerate}
\usepackage{paralist}
\usepackage{color}
\usepackage[colorlinks=true,linkcolor=blue,citecolor=blue,pdfborder={0 0 0}]{hyperref}
\usepackage{booktabs,rotating}
\usepackage[nolists]{endfloat}
\usepackage{bm}

\numberwithin{equation}{section}

\newcommand{\eps}{\varepsilon}
\newcommand{\N}{\mathbb{N}}
\newcommand{\R}{\mathbb{R}}

\newcommand{\dd}{\mathrm{d}}
\newcommand{\A}{\mathbb{A}}

\newcommand{\FF}{\mathcal{F}}
\newcommand{\LL}{\mathcal{L}}
\newcommand{\I}{\mathcal{I}}
\newcommand{\B}{\mathbb{B}}
\newcommand{\Hb}{\mathbb{H}}

\newcommand{\U}{\mathbb{U}}
\newcommand{\Z}{\mathbb{Z}}
\renewcommand{\S}{\mathbb{S}}
\newcommand{\T}{\mathbb{T}}
\newcommand{\Ex}{\operatorname{E}}
\newcommand{\var}{\operatorname{var}}
\newcommand{\cov}{\operatorname{cov}}
\newcommand{\MSE}{\mathrm{MSE}}
\newcommand{\1}{\mathbf{1}}
\newcommand{\ip}[1]{\lfloor #1 \rfloor}
\renewcommand{\vec}{\bm}
\newcommand{\pobs}[1]{\hat{\bm #1}}
\renewcommand{\Pr}{\mathrm{P}}
\newcommand{\p}{\overset{\Pr}{\to}}
\newcommand{\as}{\overset{\mathrm{a.s.}}{\longrightarrow}}

\theoremstyle{plain}
\newtheorem{prop}{Proposition}
\newtheorem{cor}[prop]{Corollary}
\newtheorem{lem}[prop]{Lemma}

\begin{document}
\maketitle

\begin{abstract} 
A class of tests for change-point detection designed to be particularly sensitive to changes in the cross-sectional rank correlation of multivariate time series is proposed. The derived procedures are based on several multivariate extensions of Spearman's rho. Two approaches to carry out the tests are studied: the first one is based on resampling, the second one consists of estimating the asymptotic null distribution. The asymptotic validity of both techniques is proved under the null for strongly mixing observations. A procedure for estimating a key bandwidth parameter involved in both approaches is proposed, making the derived tests parameter-free. Their finite-sample behavior is investigated through Monte Carlo experiments. Practical recommendations are made and an illustration on trivariate financial data is finally presented.

\medskip

\noindent {\it Keywords:} change-point detection; empirical copula; HAC kernel variance estimator; multiplier central limit theorems; partial-sum processes; ranks; Spearman's rho; strong mixing.
\end{abstract}


\section{Introduction}

Let $\vec X_1,\dots,\vec X_n$ be a multivariate times series of $d$-dimensional observations and, for any $i \in \{1,\dots,n\}$, let $F^{(i)}$ denote the cumulative distribution function (c.d.f.) of $\vec X_i$. We are interested in procedures for testing $H_0: F^{(1)} = \dots = F^{(n)}$ against $\neg H_0$. Notice that the aforementioned null hypothesis can be simply rewritten as
\begin{equation}
\label{H0}
  H_0 : \,\exists \, F \text{ such that } 
  \vec X_1, \ldots, \vec X_n \text{ have c.d.f. } F.
\end{equation}
Such statistical procedures are commonly referred to as {\em tests for change-point detection} \citep[see, e.g.,][for an overview of possible approaches]{CsoHor97}. The majority of tests for $H_0$ developed in the literature deal with the case $d=1$. We aim at developing {\em nonparametric} tests for {\em multivariate} time series that are particularly sensitive to changes in the {\em dependence} among the components of the $d$-dimensional observations. The availability of such tests seems to be of great practical importance for the analysis of economic data, among others. In particular, assessing whether the dependence among financial assets can be considered constant or not over a given time period appears crucial for risk management, portfolio optimization and related statistical modeling \citep[see, e.g.,][and the references therein for a more detailed discussion about the motivation for such statistical procedures]{WieDehvanVog14,DehVogWenWie14}.

The above context, rather naturally, suggests to address the informal notion of {\em dependence} through that of {\em copula} \citep[see, e.g.,][]{Nel06}. Assume that $H_0$ in~\eqref{H0} holds and that, additionally, the common marginal c.d.f.s $F_1,\dots,F_d$ of $\vec X_1,\dots,\vec X_n$ are continuous. Then, from the work of \cite{Skl59}, the common multivariate c.d.f.\ $F$ of the observations can be written as
$$
F(\vec x) = C \{ F_1(x_1),\dots,F_d(x_d) \}, \qquad \vec x \in \R^d,
$$
where the function $C:[0, 1]^d \to [0,1]$ is the unique {\em copula} associated with $F$. It follows that $H_0$ can be rewritten as $H_{0,m} \cap H_{0,c}$, where
\begin{align}
\label{H0m}
  H_{0,m} &: \,\exists \, F_1,\dots, F_d \text{ such that } 
  \vec X_1, \ldots, \vec X_n \text{ have marginal c.d.f.s } F_1, \dots, F_d,
\\
\label{H0c}
  H_{0,c} &: \,\exists \, C \text{ such that } 
  \vec X_1, \ldots, \vec X_n \text{ have copula } C.
\end{align}

Several nonparametric tests designed to be particularly sensitive to certain alternatives under $H_{0,m} \cap  \neg H_{0,c}$ were proposed in the literature. Tests for the constancy of Kendall's tau (which is a functional of $C$) were investigated by \cite{GomHor99} \citep[see also][]{GomHor02} and \cite{QueSaiFav13} in the case of serially independent observations. A version of the previous tests adapted to a very general class of bivariate time series was proposed by \cite{DehVogWenWie14}. Recent multivariate alternatives are the tests studied in \citet[see also the references therein]{BucKojRohSeg14} based on Cram\'er--von Mises functionals of the {\em sequential empirical copula process}.

The aim of this work is to derive tests for the constancy of several multivariate extensions of Spearman's rho (which are also functionals of $C$) in multivariate strongly mixing time series. A similar problem was recently tackled by \cite{WieDehvanVog14}. However, as the functional they considered does not exactly correspond to a multivariate extension of Spearman's rho (because of the way ranks are calculated), the corresponding test turn out to have a rather low power. We remedy to that situation by computing ranks with respect to the relevant subsamples. From a theoretical perspective, as in \cite{WieDehvanVog14}, no assumptions on the first order partial derivatives of the copula are made. The latter is actually an advantage of the studied tests over that investigated in \cite{BucKojRohSeg14}. An inconvenience with respect to the aforementioned approach is however that, as all tests based on moments of copulas (such as Spearman's rho or Kendall's tau), the derived tests will have no power, by construction, against alternatives involving changes in the copula at a constant value of Spearman's rho.

To carry out the tests, we propose two approaches for computing approximate p-values: the first one is based on resampling while the second one consists of estimating the asymptotic null distribution. In addition, a procedure for estimating a key bandwidth parameter involved in both approaches is proposed, making the derived tests fully data-driven. The versions of the studied tests based on the estimation of the asymptotic null distribution can be seen as alternatives to the test based on Kendall's tau recently proposed by \cite{DehVogWenWie14}. 

The paper is organized as follows. The test statistics are defined in the second section and their limiting null distribution is established under strong mixing. Section~\ref{approx_pval} presents two approaches for computing approximate p-values based, respectively, on bootstrapping and on the estimation of an asymptotic variance. The fourth section partially reports the results of Monte Carlo experiments involving bivariate and fourvariate time series generated from autoregressive and GARCH-like models. The fifth section contains practical recommendations and an illustration on trivariate financial data, while the last section concludes.

In the rest of the paper, the arrow~`$\leadsto$' denotes weak convergence in the sense of Definition~1.3.3 in \cite{vanWel96}. Also, given a set $T$, $\ell^\infty(T;\R)$ denotes the space of all bounded real-valued functions on $T$ equipped with the uniform metric. The proofs of the stated theoretical results are available in the online supplementary material and the studied tests for change-point detection are implemented in the package {\tt npcp} \citep{npcp} for the \textsf{R} statistical system \citep{Rsystem}.

\section{Test statistics}

\subsection{Multivariate extensions of Spearman's rho and their estimation}

Spearman's rho is a very well-known measure of bivariate dependence \citep[see, e.g.,][Section 5.1 and the references therein]{Nel06}. For a bivariate random vector with continuous margins and copula $C$, it can be expressed as 
$$
\rho(C) = 12 \int_{[0,1]^2} C(\vec u) \dd \vec u - 3 = 12 \int_{[0,1]^2} u_1 u_2 \dd C(\vec u) - 3.
$$
When the random vector of interest is $d$-dimensional with $d>2$, the following three possible extensions were proposed by \cite{SchSch07a}:
\begin{align*}
\rho_1(C) &=  \frac{d+1}{2^d - d - 1} \left\{ 2^d \int_{[0,1]^d} C(\vec u) \dd \vec u - 1 \right\}, \\
\rho_2(C) &= \rho_1(\bar C), \\
\rho_3(C) &= { d \choose 2 }^{-1} \sum_{1 \leq i < j \leq d} \rho(C^{(i,j)}),
\end{align*}
where $C^{(i,j)}$ is the bivariate margin obtained from $C$ by keeping dimensions $i$ and $j$, and $\bar C$ is the survival function corresponding to $C$. It is well-known that the latter can be expressed in terms of $C$. To see this, let $D = \{1,\dots,d\}$ and, for any $\vec u \in [0,1]^d$ and $A \subseteq D$, let $\vec u^A$ be the vector of $[0,1]^d$ such that $u^A_i = u_i$ if $i \in A$ and  $u^A_i = 1$ otherwise. Then, for any $\vec u \in [0,1]^d$, $\bar C(\vec u) = \sum_{A \subseteq D} (-1)^{|A|} C(\vec u^A)$. Other related $d$-dimensional coefficients are considered in \cite{Que09}.

Let us now discuss the estimation of the above theoretical quantities. Specifically, we assume that we have at hand $n$ copies $\vec X_1, \dots, \vec X_n$ of a $d$-dimensional random vector $\vec X$ with copula $C$ and continuous margins. Given an estimator of $C$, natural estimators of $\rho_1(C)$, $\rho_2(C)$ and $\rho_3(C)$ can be obtained using the plug-in principle. Restricting attention to a sample $\vec X_k, \dots, \vec X_l$, $1 \leq k \leq l \leq n$, for reasons that will become clear in the next subsection, a natural estimator of $C$ is given by 
\begin{equation}
\label{eq:Ckl}
  C_{k:l}(\vec{u}) 
  = \frac{1}{l-k+1} \sum_{i=k}^l \1(\pobs{U}_i^{k:l} \leq \vec{u}), \qquad \vec{u} \in [0, 1]^d,
\end{equation}
where
\begin{equation}
\label{eq:pseudo}
  \pobs{U}_i^{k:l} 
  = \frac{1}{l-k+1} (R_{i1}^{k:l},\dots,R_{id}^{k:l}), \qquad 
  i \in \{ k, \dots, l \},
\end{equation}
with $R_{ij}^{k:l} = \sum_{t=k}^l \1( X_{tj} \le X_{ij} )$ the maximal rank of $X_{ij}$ among $X_{kj},\dots,X_{lj}$. The quantity given by~\eqref{eq:Ckl} is commonly referred to as the {\em empirical copula} of $\vec X_k, \dots, \vec X_l$ \citep[see, e.g.,][]{Rus76,Deh81}. Corresponding natural estimators of the three aforementioned multivariate versions of Spearman's rho are therefore $\rho_1(C_{k:l})$, $\rho_2(C_{k:l})$ and $\rho_3(C_{k:l})$, respectively. 


It is important to notice that we do not necessarily assume the observations to be serially independent. Serial independence {\em and} continuity of the marginal distributions together guarantee the absence of ties in the $d$ component series. However, continuity of the marginal distributions alone is \emph{not} sufficient to guarantee the absence of ties when the observations are serially dependent \cite[see, e.g.,][Example~4.2]{BucSeg14}. This is the reason why maximal ranks are used in~\eqref{eq:pseudo}. The possible presence of ties in the component series makes the study of the tests under consideration substantially more complicated. 

\subsection{Change-point statistics}
\label{sec:cpstat}

To derive tests for change-point detection particularly sensitive to changes in the strength of the cross-sectional dependence, one natural possibility is to base these tests on differences of Spearman's rhos. By analogy with the classical approach to change-point analysis \citep[see, e.g.,][]{CsoHor97}, one could for instance consider the following three test statistics: 
\begin{equation}
\label{eq:Sni}
S_{n,i} = \max_{1 \leq k \leq n-1} \frac{k(n-k)}{n^{3/2}} \left| \rho_i(C_{1:k}) -  \rho_i(C_{k+1:n}) \right|, \qquad i \in \{1,2,3\},
\end{equation}
where $C_{1:k}$ and $C_{k+1:n}$ are the empirical copulas of the subsamples $\vec X_1,\dots,\vec X_k$ and $\vec X_{k+1},\dots,\vec X_n$, respectively, defined analogously to~\eqref{eq:Ckl}. All three statistics above turn out to be particular cases of a generic statistic which is the primary focus of this work. Before we can define it, some additional notation is necessary. 

For any $A \subseteq D = \{1,\dots,d\}$, let $\phi_A$ be the map from $\ell^\infty([0,1]^d;\R)$ to $\R$ defined by 
\begin{equation}
\label{eq:phiA}
\phi_A(g) = \int_{[0,1]^d} g(\vec u^A) \dd \vec u, \qquad g \in \ell^\infty([0,1]^d;\R).
\end{equation}
Then, define the empirical process 
$$
\T_{n,A}(s) =  \sqrt{n} \, \lambda_n(0,s) \, \lambda_n(s,1) \, \{ \phi_A(C_{1:\ip{ns}}) - \phi_A(C_{\ip{ns}+1:n}) \}, \qquad s \in [0,1],
$$
where $\lambda_n(s, t) = (\ip{nt}-\ip{ns})/n$ for $(s,t) \in \Delta = \{ (s, t) \in [0, 1]^2: s \le t \}$, and with the additional convention that $C_{k:l} = 0$ whenever $k > l$. Simple calculations reveal that $\T_{n,\emptyset} = 0$. Next, consider the $\R^{2^d - 1}$-valued empirical process
\begin{equation}
\label{eq:Tn}
\T_n(s) = ( \T_{n,\{1\}}(s), \T_{n,\{2\}}(s), \dots, \T_{n,D}(s) ), \qquad s \in [0,1].
\end{equation}
Finally, given a function $f:\R^{2^d - 1} \to \R$, define the generic change-point statistic
\begin{equation}
\label{eq:Snf}
S_{n,f} = \sup_{s \in [0,1]} |f\{\T_n(s)\}| = \max_{1 \leq k \leq n-1} |f\{\T_n(k/n)\}|.
\end{equation}

We shall now verify that the statistics $S_{n,i}$, $i \in \{1,2,3\}$, given by~\eqref{eq:Sni} are particular cases of $S_{n,f}$ when $f$ is linear, that is, when there exists a vector $\vec a \in \R^{2^d - 1}$ such that, for any $\vec x \in \R^{2^d - 1}$, $f(\vec x) = \vec a^\top \vec x$. As we continue, with some abuse of notation, we index the components of vectors of $\R^{2^d - 1}$ by subsets of $D$ of cardinality greater than 1, i.e., for any $\vec x \in \R^{2^d - 1}$, we write $\vec x = ( x_{\{1\}}, x_{\{2\}}, \dots, x_D )$. Then, we have $S_{n,i} = S_{n,f_i}$, $i \in \{1,2,3\}$, where, for any $\vec x \in \R^{2^d - 1}$,
\begin{gather*}
f_1(\vec x) = \frac{(d+1)2^d}{2^d - d - 1} x_{D}, \quad f_2(\vec x) = \frac{(d+1)2^d}{2^d - d - 1} \sum_{A \subseteq D \atop |A| \geq 1} (-1)^{|A|} x_A, \\ f_3(\vec x) = \frac{24}{d(d-1)} \sum_{A \subseteq D \atop |A| = 2} x_A.
\end{gather*}
Similar relationships hold for the statistics constructed from the additional coefficients mentioned in \cite{Que09}, though the corresponding functions $f$ are not necessarily linear anymore but only continuous. 

Let us make a brief remark concerning the statistic $S_{n,2}$. Expressing it as $S_{n,f_2}$ above is clearly not the most efficient way to compute it. To see this, for any $1 \leq k \leq l \leq n$, define
$$
\bar C_{k:l}(\vec{u}) = \frac{1}{l-k+1} \sum_{i=k}^l \1(\pobs{U}_i^{k:l} > \vec{u}), \qquad \vec u \in [0,1]^d,
$$
where the $\pobs{U}_i^{k:l}$ are defined in~\eqref{eq:pseudo}, and notice that, for any $\vec u \in [0,1]^d$, $\bar C_{k:l}(\vec u) = \sum_{A \subseteq D} (-1)^{|A|} C_{k:l}(\vec u^A)$, where $C_{k:l}$ is defined in~\eqref{eq:Ckl}. Then, by definition of $\rho_2$,
$$
S_{n,2} = \max_{1 \leq k \leq n-1} \frac{k(n-k)}{n^{3/2}} \left| \rho_1(\bar C_{1:k}) -  \rho_1(\bar C_{k+1:n}) \right|.
$$
Under the assumption of no ties in the $d$ component series, some additional simple calculations reveal that the latter is actually nothing else than $S_{n,1}$ computed from the sample $-\vec X_1,\dots,-\vec X_n$.

We end this section by a discussion of the differences between $S_{n,1}$ and the similar statistic considered in \cite{WieDehvanVog14}. Instead of basing their approach on the empirical copula, these authors considered the alternative estimator of $C$ defined, for any $1 \leq k \leq l \leq n$, as
\begin{equation}
\label{eq:Ckln}
  C_{k:l,n}(\vec{u}) 
  = \frac{1}{l-k+1} \sum_{i=k}^l \1(\pobs{U}_i^{1:n} \leq \vec{u}),
  \qquad \vec{u} \in [0, 1]^d,
\end{equation}
with the convention that $C_{k:l,n}=0$ if $k > l$. The apparently subtle yet crucial difference between $C_{k:l}$ in~\eqref{eq:Ckl} and $C_{k:l,n}$ above is that the scaled ranks are computed relative to the complete sample $\vec X_1, \ldots, \vec X_n$ for $C_{k:l,n}$, while, for $C_{k:l}$, they are computed relative to the subsample $\vec X_k, \ldots, \vec X_l$. As a consequence, the analogue of the statistic $S_{n,1}$ considered in \cite{WieDehvanVog14} is not really a maximally selected absolute difference of sample Spearman's rhos. From a practical perspective, as illustrated empirically in \cite{BucKojRohSeg14}, the use of $C_{k:l}$ instead of $C_{k:l,n}$ in a change-point detection framework results in tests that are more powerful when the change in distribution in only due to a change in the copula. We provide similar empirical evidence in Section~\ref{sims}: tests based on $S_{n,1}$ appear substantially more powerful than their analogues based on~\eqref{eq:Ckln} for alternatives involving a change of $\rho_1(C)$ at constant margins. Reasons that explain this improved efficiency are discussed in \citet[Section 2]{BucKojRohSeg14}.

\subsection{Limiting null distribution under strong mixing}

Let us first recall the notion of {\em strongly mixing sequence}. For a sequence of $d$-dimensional random vectors $(\vec Y_i)_{i \in \Z}$, the $\sigma$-field generated by $(\vec Y_i)_{a \leq i \leq b}$, $a, b \in \Z \cup \{-\infty,+\infty \}$, is denoted by $\FF_a^b$. The strong mixing coefficients corresponding to the sequence $(\vec Y_i)_{i \in \Z}$ are defined by
$$
\alpha_r = \sup_{p \in \Z} \sup_{A \in \FF_{-\infty}^p,B\in \FF_{p+r}^{+\infty}} | P(A \cap B) - P(A) P(B) |
$$
for strictly positive integer $r$. The sequence $(\vec Y_i)_{i \in \Z}$ is said to be \emph{strongly mixing} if $\alpha_r \to 0$ as $r \to \infty$. 

The limiting null distribution of the vector-valued empirical process $\T_n$ defined in~\eqref{eq:Tn} can be obtained by rewriting its components in terms of the processes 
\begin{equation}
\label{eq:SnA}
\S_{n,A}(s,t) =  \sqrt{n} \lambda_n(s,t) \{ \phi_A(C_{\ip{ns}+1:\ip{nt}}) - \phi_A(C) \}, \qquad (s,t) \in \Delta,
\end{equation}
for $A \subseteq D, |A| \geq 1$. Indeed, it is easy to verify that, under $H_0$ defined in~\eqref{H0},
\begin{equation}
\label{eq:TnAsH0}
\T_{n,A}(s) = \lambda_n(s,1) \S_{n,A}(0,s) -  \lambda_n(0,s) \S_{n,A}(s,1), \qquad s \in [0,1].
\end{equation}
As we shall see below, the limiting null distribution of $\T_n$ is then a mere consequence of the fact that the empirical processes $\S_{n,A}$, $A \subseteq D$, $|A| \geq 1$, are asymptotically equivalent to continuous functionals of the sequential empirical process 
\begin{equation}
\label{eq:seqep}
  \B_n(s,t, \vec u) = \frac{1}{\sqrt{n}} \sum_{i=\ip{ns} + 1}^{\ip{nt}} \{\1(\vec U_i \leq \vec u) - C(\vec u) \}, \qquad (s,t, \vec u) \in \Delta \times [0, 1]^d,
\end{equation}
where $\vec U_1,\dots,\vec U_n$ is the unobservable sample obtained from $\vec X_1, \dots, \vec X_n$ by the probability integral transforms $U_{ij} = F_j(X_{ij})$, $i \in \{1,\dots,n\}$, $j \in D$. 

If $\vec U_1,\dots,\vec U_n$ is drawn from a strictly stationary sequence $(\vec U_i)_{i \in \Z}$ whose strong mixing coefficients satisfy $\alpha_r = O(r^{-a})$ with $a > 1$, we have from \cite{Buc14} that $\B_n(0,\cdot,\cdot)$ converges weakly in $\ell^\infty([0,1]^{d+1};\R)$ to a tight centered Gaussian process $\B_C^\circ$ with covariance function $\cov\{\B_C^\circ(s,\vec u), \B_C^\circ(t,\vec v)\} = (s \wedge t) \kappa_C(\vec u,\vec v)$, $(s,\vec u), (t, \vec v) \in [0,1]^{d+1}$, where 
\begin{equation}
\label{eq:kappaC}
\kappa_C(\vec u,\vec v) = \cov\{\B_C^\circ(1,\vec u), \B_C^\circ(1,\vec v)\} = \sum_{k \in \Z} \cov\{ \1(\vec U_0 \leq \vec u), \1(\vec U_k \leq \vec v) \}.
\end{equation}
As a consequence of the continuous mapping theorem, $\B_n \leadsto \B_C$ in $\ell^\infty(\Delta \times [0, 1]^d;\R)$, where
\begin{equation}
\label{eq:BC}
  \B_C(s, t,\vec u) = \B_C^\circ(t,\vec u) - \B_C^\circ(s,\vec u), \qquad (s, t,\vec u) \in \Delta \times [0,1]^d.
\end{equation}

The following proposition, proved in Section~\ref{proof:prop:weak_SnA_sm} of the supplementary material, is the key step for obtaining the limiting null distribution of the vector-valued process $\T_n$ defined in~\eqref{eq:Tn}.

\begin{prop}
\label{prop:weak_SnA_sm}
Assume that $\vec X_1,\dots,\vec X_n$ is drawn from a strictly stationary sequence $(\vec X_i)_{i \in \Z}$ with continuous margins and whose strong mixing coefficients satisfy $\alpha_r = O(r^{-a})$, $a > 1$. Then, for any $A \subseteq D$, $|A| \geq 1$,
\begin{equation}
\label{eq:asymequivSnA}
\sup_{(s,t) \in \Delta} | \S_{n,A}(s,t) - \psi_{C,A} \{ \B_n(s,t,\cdot) \} | = o_\Pr(1),
\end{equation}
where $\psi_{C,A}$ is a linear map from $\ell^\infty([0,1]^d;\R)$ to $\R$ defined by
\begin{equation}
\label{eq:psiCA}
\psi_{C,A}(g) = \phi_A(g) - \int_{[0,1]^d} \sum_{j \in A} \prod_{l \in A \setminus \{j\}} (1-v_l) g(\vec v^{\{j\}}) \dd C(\vec v), \qquad g \in \ell^\infty([0,1]^d;\R),
\end{equation}
with $\phi_A$ given in~\eqref{eq:phiA}.
\end{prop}

From the work of \cite{Mok88}, we know that the strong mixing conditions stated in the previous proposition (as well as those stated in the forthcoming propositions and corollaries) are for instance satisfied (with much to spare) when $\vec X_1,\dots,\vec X_n$ is drawn from a stationary vector ARMA process with absolutely continuous innovations. A similar conclusion holds for a large class of GARCH processes \cite[see][Section~5, and the references therein]{Lin09}.

The next result, proved in Section~\ref{proof:cor:weak_TnA} of the supplementary material, is a consequence of the previous proposition and establishes the limiting null distribution of the generic statistic $S_{n,f}$ defined in~\eqref{eq:Snf} under strong mixing.

\begin{cor}
\label{cor:weak_TnA}
Under the conditions of Proposition~\ref{prop:weak_SnA_sm}, 
\begin{equation}
\label{eq:wcTn}
\T_n \leadsto s \mapsto \T_C(s) = \left( \T_{C,\{1\}}(s), \T_{C,\{2\}}(s), \dots, \T_{C,D}(s) \right)
\end{equation}
in $\ell^\infty([0,1];\R^{2^d-1})$, where
\begin{equation}
\label{eq:TC}
\T_C(s) = \psi_C\{ \B_C(0,s,\cdot) - s \B_C(0,1,\cdot) \}, \qquad s \in [0,1], 
\end{equation}
with $\B_C$ defined in~\eqref{eq:BC} and $\psi_C$ a map from $\ell^\infty([0,1]^d;\R)$ to $\R^{2^d - 1}$ defined by
\begin{equation}
\label{eq:psiC}
\psi_C(g) = \left( \psi_{C,\{1\}}(g), \psi_{C,\{2\}}(g), \dots, \psi_{C,D}(g) \right), \qquad g \in \ell^\infty([0,1]^d;\R).
\end{equation}
As a consequence, for any $f:\R^{2^d - 1} \to \R$ continuous, 
$$
S_{n,f} = \sup_{s \in [0,1]} | f \{ \T_n(s) \} | \leadsto S_{C,f} = \sup_{s \in [0,1]} | f \{ \T_C(s) \} |,
$$
and, if $f$ is additionally linear and $\sigma_{C,f}^2 = \var[ f \circ \psi_C\{\B_C(0,1,\cdot)\} ] > 0$, the weak limit of $\sigma_{C,f}^{-1} S_{n,f}$ is equal in distribution to $\sup_{s \in [0,1]} |\U(s)|$, where $\U$ is a standard Brownian bridge on $[0,1]$.
\end{cor}

\section{Computation of approximate p-values}
\label{approx_pval}

Corollary~\ref{cor:weak_TnA} suggests two related ways to compute p-values for the generic test statistic $S_{n,f}$ defined in~\eqref{eq:Snf}. The first approach, based on resampling, consists of exploiting the fact that, under $H_0$, $\T_n$ defined in~\eqref{eq:Tn} is asymptotically equivalent to a continuous functional of the sequential empirical process $\B_n$ defined in~\eqref{eq:seqep} and can be applied as soon as $f:\R^{2^d - 1} \to \R$ is continuous. The second approach, restricted to the situation when $f$ is linear, is motivated by the last claim of Corollary~\ref{cor:weak_TnA}. It consists of estimating $\sigma_{C,f}^2$ and thus the asymptotic null distribution of~$S_{n,f}$.

\subsection{Approximate p-values by bootstrapping}
\label{sec:resampling}

The first approach that we consider consists of bootstrapping the vector-valued empirical process $\T_n$ defined in~\eqref{eq:Tn} using a bootstrap for the sequential empirical process $\B_n$. This way of proceeding actually allows us to consider not only linear but also {\em continuous} functions~$f$ in~\eqref{eq:Snf}. More specifically, we consider a {\em multiplier bootstrap} for $\B_n$ in the spirit of \citet[Chapter 2.9]{vanWel96} when observations are serially independent, or \citet[Section 3.3]{Buh93} when they are serially dependent. In the latter case, we rely on the recent work of \cite{BucKoj14}.

The notion of {\em multiplier sequence} is central to this resampling technique. We say that a sequence of random variables $(\xi_{i,n})_{i \in \Z}$ is an \emph{i.i.d.\ multiplier sequence} if:
\begin{enumerate}[({M}0)]
\item 
$(\xi_{i,n})_{i \in \Z}$ is i.i.d., independent of $\vec{X}_1, \ldots, \vec{X}_n$, with distribution not changing with~$n$, having mean 0, variance 1, and being such that $\int_0^\infty \{ \Pr(|\xi_{0,n}| > x) \}^{1/2} \dd x < \infty$.
\end{enumerate}
We say that a sequence of random variables $(\xi_{i,n})_{i \in \Z}$ is a \emph{dependent multiplier sequence} if:
\begin{enumerate}[({M}1)]
\item 
The sequence $(\xi_{i,n})_{i \in \Z}$ is strictly stationary with $\Ex(\xi_{0,n}) = 0$, $\Ex(\xi_{0,n}^2) = 1$ and $\sup_{n \geq 1} \Ex(|\xi_{0,n}|^\nu) < \infty$ for all $\nu \geq 1$, and is independent of the available sample $\vec X_1,\dots,\vec X_n$.
\item 
There exists a sequence $\ell_n \to \infty$ of strictly positive constants such that $\ell_n = o(n)$ and the sequence $(\xi_{i,n})_{i \in \Z}$ is $\ell_n$-dependent, i.e., $\xi_{i,n}$ is independent of $\xi_{i+h,n}$ for all $h > \ell_n$ and $i \in \N$. 
\item 
There exists a function $\varphi:\R \to [0,1]$, symmetric around 0, continuous at $0$, satisfying $\varphi(0)=1$ and $\varphi(x)=0$ for all $|x| > 1$ such that $\Ex(\xi_{0,n} \xi_{h,n}) = \varphi(h/\ell_n)$ for all $h \in \Z$.
\end{enumerate}
The choice of the function $\varphi$ and an approach to generate dependent multiplier sequences is briefly discussed in Section~\ref{sims}. More details can be found in \citet[Section~5.2]{BucKoj14}.

Let $M$ be a large integer and let $(\xi_{i,n}^{(1)})_{i \in \Z},\dots,(\xi_{i,n}^{(M)})_{i \in \Z}$ be $M$ independent copies of the same multiplier sequence. Then, following \cite{BucKoj14} and \cite{BucKojRohSeg14}, for any $m \in \{1,\dots,M\}$ and $(s,t,\vec u) \in \Delta \times [0,1]^d$, let
\begin{align}
\nonumber
\hat{\B}_n^{(m)}(s, t,\vec{u}) &= \frac{1}{\sqrt{n}} \sum_{i=\ip{ns}+1}^{\ip{nt}} \xi_{i,n}^{(m)} \{ \1 ( \pobs{U}_i^{1:n} \leq \vec{u} ) - C_{1:n}(\vec{u}) \}, \\
\label{eq:checkBnm}
\check{\B}_n^{(m)}(s, t,\vec{u}) &= \frac{1}{\sqrt{n}} \sum_{i=\ip{ns}+1}^{\ip{nt}} ( \xi_{i,n}^{(m)} - \bar  \xi_{\ip{ns}+1:\ip{nt}}^{(m)} ) \1 ( \pobs{U}_i^{\ip{ns}+1:\ip{nt}} \leq \vec{u} ),
\end{align}
where $\bar \xi_{k:l}^{(m)}$ is the arithmetic mean of $\xi_{i,n}^{(m)}$ for $i \in \{ k, \ldots, l \}$.

The following proposition is a consequence of Theorem~1 in \cite{HolKojQue13}, Theorem~2.1 and the proof of Proposition 4.2 in \cite{BucKoj14}, as well as the proof of Proposition~4.3 in \cite{BucKojRohSeg14}. It suggests interpreting the multiplier replicates $\hat \B_n^{(1)},\dots,\hat \B_n^{(M)}$ (resp.\ $\check \B_n^{(1)},\dots,\check \B_n^{(M)}$) as ``almost'' independent copies of $\B_n$ as $n$ increases.

\begin{prop}
\label{prop:multBn}
Assume that either
\begin{enumerate}[\bf (i)]
\item the random vectors $\vec X_1,\dots,\vec X_n$ are i.i.d.\ with continuous margins and the sequences $(\xi_{i,n}^{(1)})_{i \in \Z},\dots,(\xi_{i,n}^{(M)})_{i \in \Z}$ are independent copies of a multiplier sequence satisfying~(M0),
\item or the random vectors $\vec X_1,\dots,\vec X_n$ are drawn from a strictly stationary sequence $(\vec X_i)_{i \in \Z}$ with continuous margins whose strong mixing coefficients satisfy $\alpha_r = O(r^{-a})$ for some $a > 3 + 3d/2$, and $(\xi_{i,n}^{(1)})_{i \in \Z}$, \dots, $(\xi_{i,n}^{(M)})_{i \in \Z}$ are independent copies of a dependent multiplier sequence satisfying~(M1)--(M3) with $\ell_n = O(n^{1/2 - \eps})$ for some $0 < \eps < 1/2$. 
\end{enumerate}
Then,
\begin{align*}
\left(\B_n,\hat \B_n^{(1)},\dots,\hat \B_n^{(M)} \right) &\leadsto \left(\B_C,\B_C^{(1)},\dots,\B_C^{(M)} \right), \\
\left(\B_n,\check \B_n^{(1)},\dots,\check \B_n^{(M)} \right) &\leadsto \left(\B_C,\B_C^{(1)},\dots,\B_C^{(M)} \right)
\end{align*}
in $\{\ell^\infty(\Delta \times [0,1]^d;\R)\}^{M+1}$, where $\B_C$ is given in~\eqref{eq:BC} and $\B_C^{(1)},\dots,\B_C^{(M)}$ are independent copies of $\B_C$. 
\end{prop}

Starting from the quantities defined above, we shall now define appropriate multiplier replicates under $H_0$ of $\T_n$ defined in~\eqref{eq:Tn}. From~\eqref{eq:TnAsH0}, we see that to do so, we first need to define multiplier replicates of the processes $\S_{n,A}$, $A \subseteq D$, $|A| \geq 1$, defined in~\eqref{eq:SnA}. From~\eqref{eq:asymequivSnA} and Proposition~\ref{prop:multBn}, natural candidates would be the processes $(s,t) \mapsto \psi_{C,A}\{\hat \B_n^{(m)}(s,t,\cdot)\}$ or the processes $(s,t) \mapsto \psi_{C,A}\{\check \B_n^{(m)}(s,t,\cdot)\}$, $m \in \{1,\dots,M\}$, where the map $\psi_{C,A}$ is defined in~\eqref{eq:psiCA}. These however still depend on the unknown copula~$C$. The latter could be estimated either by $C_{1:n}$ or by $C_{\ip{ns}+1:\ip{nt}}$, which led us to consider the following two computable versions instead:
$$
\hat{\S}_{n,A}^{(m)}(s,t) = \psi_{C_{1:n},A}\{\hat \B_n^{(m)}(s,t,\cdot)\}, \quad \check{\S}_{n,A}^{(m)}(s,t) = \psi_{C_{\ip{ns}+1:\ip{nt}},A}\{\check \B_n^{(m)}(s,t,\cdot)\}, 
$$
for $(s,t) \in \Delta$. The processes $\check{\S}_{n,A}^{(m)}$ were found to lead to better behaved tests than the $\hat{\S}_{n,A}^{(m)}$ in our Monte Carlo experiments, which is why, from now on, we focus solely on the former. It is easy to verify that the $\check{\S}_{n,A}^{(m)}$ can be rewritten as
$$
\check{\S}_{n,A}^{(m)}(s,t) = \frac{1}{\sqrt{n}} \sum_{i=\ip{ns}+1}^{\ip{nt}}( \xi_{i,n}^{(m)} - \bar  \xi_{\ip{ns}+1:\ip{nt}}^{(m)} ) \I_{C_{\ip{ns}+1:\ip{nt}},A} (\pobs{U}_i^{\ip{ns}+1:\ip{nt}}),
$$
where, for any $\vec u \in [0,1]^d$,
\begin{align}
\nonumber
\I_{C,A}(\vec u) &= \psi_{C,A}\{\1(\vec u \leq \cdot)\} \\ 
\label{eq:ICA}
&= \prod_{l \in A}(1-u_l) - \int_{[0,1]^d} \sum_{j \in A} \prod_{l \in A \setminus \{j\}} (1-v_l) \1(u_j\leq v_j) \dd C(\vec v).
\end{align}
Next, by analogy with~\eqref{eq:TnAsH0}, for any $m \in \{1,\dots,M\}$, $A \subseteq D$, $|A| \geq 1$, let
$$
\check \T_{n,A}^{(m)}(s) = \lambda_n(s,1) \check{\S}_{n,A}^{(m)}(0,s) -  \lambda_n(0,s) \check{\S}_{n,A}^{(m)}(s,1), \qquad s \in [0,1],
$$
and let $\check \T_n^{(m)}$ be the corresponding version of $\T_n$ in~\eqref{eq:Tn}. Finally, for some continuous function $f:\R^{2^d - 1} \to \R$, let $\check S_{n,f}^{(m)} = \sup_{s \in [0,1]} | f\{\check \T_n^{(m)}(s)\} |$ by analogy with~\eqref{eq:Snf}. Interpreting the $\check S_{n,f}^{(m)}$ as multiplier replicates of $S_{n,f}$ under $H_0$, it is natural to compute an approximate p-value for the test as 
\begin{equation}
\label{eq:pval}
\frac{1}{M} \sum_{m=1}^M \1 \left( \check S_{n,f}^{(m)} \geq S_{n,f} \right).
\end{equation}
The null hypothesis is rejected if the estimated p-value is smaller than the desired significance level. 

The following result, proved in Section~\ref{proof:prop:multTn} of the supplementary material, can be combined with Proposition F.1 in \cite{BucKoj14} to show that a test based on $S_{n,f}$ whose p-value is computed as in~\eqref{eq:pval} will hold its level asymptotically as $n \to \infty$ followed by $M \to \infty$.

\begin{prop}
\label{prop:multTn}
Under the conditions of Proposition~\ref{prop:multBn}, for any $A \subseteq D$, $|A| \geq 1$, 
\begin{align*}
\left(\S_{n,A}, \check \S_{n,A}^{(1)}, \dots, \check \S_{n,A}^{(M)} \right) &\leadsto \left( \S_{C,A}, \S_{C,A}^{(1)}, \dots, \S_{C,A}^{(M)} \right)
\end{align*}
in $\{\ell^\infty(\Delta;\R)\}^{M+1}$, where, for any $(s,t) \in \Delta$, $\S_{C,A}(s,t) = \psi_{C,A}\{ \B_C(s,t,\cdot) \}$ and $\S_{C,A}^{(1)},\dots,\S_{C,A}^{(M)}$ are independent copies of $S_{C,A}$. As a consequence, 
\begin{align*}
\left(\T_n, \check \T_n^{(1)}, \dots, \check \T_n^{(M)} \right) &\leadsto \left(\T_C, \T_C^{(1)}, \dots, \T_C^{(M)} \right)
\end{align*}
in $\{\ell^\infty([0,1];\R^{2^d-1})\}^{M+1}$, where $\T_C$ is given in~\eqref{eq:TC} and $\T_C^{(1)},\dots,\T_C^{(M)}$ are independent copies of $T_C$, and, for any continuous function $f:\R^{2^d - 1} \to \R$,
\begin{align*}
\left(S_{n,f}, \check S_{n,f}^{(1)}, \dots, \check S_{n,f}^{(M)} \right) &\leadsto \left(S_{C,f}, S_{C,f}^{(1)}, \dots, S_{C,f}^{(M)} \right)
\end{align*}
in $\R^{M+1}$, where $S_{C,f} = \sup_{s \in [0,1]} | f\{\T_C(s)\}|$ and $S_{C,f}^{(1)}, \dots, S_{C,f}^{(M)}$ are independent copies of $S_{C,f}$.
\end{prop}

The finite-sample behavior of the tests under consideration based on the processes $\check{\S}_{n,A}^{(m)}$ is not however completely satisfactory: the tests appear too liberal for multivariate time series with strong cross sectional dependence. This prompted us to try other asymptotically equivalent versions of the $\check{\S}_{n,A}^{(m)}$. Under an additional assumption on the partial derivatives of the copula, the generic test statistic $S_{n,f}$ defined in~\eqref{eq:Snf} can be written under $H_0$ as a functional of the {\em two-sided sequential empirical copula process} studied in \cite{BucKoj14}, and could therefore be bootstrapped via the multiplier processes defined in~(4.4) of \cite{BucKojRohSeg14}. Without imposing any condition on the partial derivatives of the copula, the latter remark led us to consider, instead of the processes
\begin{multline}
\label{eq:checkSnA}
\check{\S}_{n,A}^{(m)}(s,t) = \phi_A \{ \check \B_n^{(m)}(s,t,\cdot) \} \\ - \int_{[0,1]^d} \sum_{j \in A} \prod_{l \in A \setminus \{j\}} (1-v_l) \check \B_n^{(m)}(s,t,\vec v^{\{j\}}) \dd C_{\ip{ns}+1:\ip{nt}}(\vec v),
\end{multline}
the processes 
\begin{multline}
\label{eq:tildeSnA}
\tilde{\S}_{n,b_n,A}^{(m)}(s,t) = \phi_A \{ \check \B_n^{(m)}(s,t,\cdot) \} \\ - \int_{[0,1]^d} \sum_{j \in A} \prod_{l \in A \setminus \{j\}} (1-v_l) \tilde \B_{n,b_n,j}^{(m)}(s,t,v_j) \dd C_{\ip{ns}+1:\ip{nt}}(\vec v),
\end{multline}
where, for any $j \in D$, $\tilde \B_{n,b_n,j}^{(m)}$ is a linearly smoothed version of $(s,t,u) \mapsto \check \B_n^{(m)}(s,t,\vec u_j)$ with $\vec u_{j}$ the vector of $[0,1]^d$ whose components are all equal to~1 except the $j$th one which is equal to $u$, and $b_n$ a strictly positive sequence of constants converging to~0. Specifically, for $(s,t,v) \in \Delta \times [0,1]$,
$$
\tilde \B_{n,b_n,j}^{(m)}(s,t,v) = \frac{1}{\sqrt{n}} \sum_{i=\ip{ns}+1}^{\ip{nt}} ( \xi_{i,n}^{(m)} - \bar  \xi_{\ip{ns}+1:\ip{nt}}^{(m)} ) \LL_{b_n}( \hat U_{ij}^{\ip{ns}+1:\ip{nt}},v ),
$$
where
$$
\LL_{b_n}( u, v) = \frac{u_+ \wedge v - u_- \wedge v}{u_+ - u-}, \qquad u,v \in [0,1],
$$
with $u_+ = (u + b_n) \wedge 1$ and $u_- = (u - b_n) \vee 0$. It is easy to verify that, for any $u \in [0,1]$, $\LL_{b_n}(u, \cdot)$ differs from $\1(u \leq \cdot)$ only on the interval $(u_-,u_+)$ on which it linearly increases from 0 to 1. 

Notice that~\eqref{eq:tildeSnA} can be rewritten as
$$
\tilde{\S}_{n,b_n,A}^{(m)}(s,t) = \frac{1}{\sqrt{n}} \sum_{i=\ip{ns}+1}^{\ip{nt}}( \xi_{i,n}^{(m)} - \bar  \xi_{\ip{ns}+1:\ip{nt}}^{(m)} ) \I_{b_n,C_{\ip{ns}+1:\ip{nt}},A} (\pobs{U}_i^{\ip{ns}+1:\ip{nt}}),
$$
where, for any $\vec u \in [0,1]^d$,
\begin{equation}
\label{eq:IbnCA}
\I_{b_n,C,A}(\vec u) = \prod_{l \in A}(1-u_l) - \int_{[0,1]^d} \sum_{j \in A} \prod_{l \in A \setminus \{j\}} (1-v_l) \LL_{b_n}(u_j,v_j) \dd C(\vec v).
\end{equation}

For any $m \in \{1,\dots,M\}$, let $\tilde \T_{n,b_n}^{(m)}$ and $\tilde S_{n,b_n,f}^{(m)}$ be the analogues of $\check \T_n^{(m)}$ and $\check S_{n,f}^{(m)}$, respectively, defined from the processes $\tilde{\S}_{n,b_n,A}^{(m)}$ in~\eqref{eq:tildeSnA}. The following result, proved in Section~\ref{proof:prop:multTn} of the supplementary material, is then the analogue of Proposition~\ref{prop:multTn} above.

\begin{prop}
\label{prop:multbarTn}
If $b_n = o(n^{-1/2})$, Proposition~\ref{prop:multTn} holds with $\check{\S}_{n,A}^{(m)}$ replaced by $\tilde{\S}_{n,b_n,A}^{(m)}$, $\check \T_n^{(m)}$ replaced by $\tilde \T_{n,b_n}^{(m)}$ and $\check S_{n,f}^{(m)}$ replaced by $\tilde S_{n,b_n,f}^{(m)}$.
\end{prop}

Finally, notice that it is possible to consider a version of the above construction in which the smoothing sequence is $b_{\ip{nt} - \ip{ns}}$ instead of $b_n$. We focused above only on the latter approach as it led to better behaved tests in our Monte Carlo experiments.

\subsection{Estimating the asymptotic null distribution}
\label{sec:asymp}

When the function $f$ used in the definition of $S_{n,f}$ in~\eqref{eq:Snf} is linear, Corollary~\ref{cor:weak_TnA} gives conditions under which, provided $\sigma_{C,f}^2 = \var[ f \circ \psi_C\{\B_C(0,1,\cdot)\} ] > 0$, the weak limit of $\sigma_{C,f}^{-1} S_{n,f}$ under $H_0$ is equal in distribution to $\sup_{s \in [0,1]} |\U(s)|$. The distribution of the latter random variable can be approximated very well (this aspect is discussed in more detail in Section~\ref{sims}). To be able to estimate an asymptotic p-value for $S_{n,f}$, it thus remains to estimate the unknown variance~$\sigma_{C,f}^2$.

Let $\Ex_\xi$ and $\var_\xi$ denote the expectation and variance, respectively, conditional on the data. By analogy with the classical way of proceeding when estimating variances using resampling procedures \citep[see, e.g.,][]{Kun89,Sha10}, in our context, a first natural estimator of the unknown variance under $H_0$ is of the form
\begin{equation}
\label{eq:checksigmaBn}
\check \sigma_{n,C,f}^2 = \var_\xi [f \circ \psi_C\{\check \B_n^{(m)}(0,1,\cdot)\}],
\end{equation}
where $\check \B_n^{(m)}$ is defined in~\eqref{eq:checkBnm}. To simplify the notation, we shall drop the superscript~$(m)$ in the rest of this section.
 The previous estimator is not computable as $C$ is unknown, which is why we will eventually consider the estimator $\check \sigma_{n,C_{1:n},f}^2$ instead. 

To obtain a more explicit expression of $\check \sigma_{n,C,f}^2$, first, let
\begin{equation}
\label{eq:IC}
\I_C(\vec u) = \left( \I_{C,\{1\}}(\vec u), \I_{C,\{2\}}(\vec u), \dots, \I_{C,D}(\vec u) \right), \qquad \vec u \in [0,1]^d,
\end{equation}
where $\I_{C,A}$, $A \subseteq D$, $|A| \geq 1$, is defined in~\eqref{eq:ICA}. From the linearity of $f \circ \psi_{C}$, we then obtain that
\begin{align*}
\check \sigma_{n,C,f}^2 &= \var_\xi \left\{ \frac{1}{\sqrt{n}} \sum_{i=1}^n (\xi_{i,n} - \bar \xi_{1:n}) f \circ \I_{C}(\pobs{U}^{1:n}_i) \right\} \\
&= \var_\xi \left[ \frac{1}{\sqrt{n}} \sum_{i=1}^n \xi_{i,n} \left\{ f \circ \I_{C}(\pobs{U}^{1:n}_i) - \frac{1}{n} \sum_{j=1}^n f \circ \I_{C}(\pobs{U}^{1:n}_j) \right\} \right]. 
\end{align*}
Using the fact that, from~\eqref{eq:ICA} and~\eqref{eq:IC}, 
$$
\frac{1}{n} \sum_{i=1}^n f \circ \I_{C}(\pobs{U}^{1:n}_i) = \frac{1}{n} \sum_{i=1}^n f \circ \psi_{C} \{ \1(\pobs{U}^{1:n}_i \leq \cdot) \} = f \circ \psi_{C}(C_{1:n}),
$$
we obtain that
\begin{multline*}
\check \sigma_{n,C,f}^2 = \frac{1}{n} \sum_{i,j=1}^n \Ex_\xi ( \xi_{i,n} \xi_{j,n} ) f \left\{ \I_{C}(\pobs{U}^{1:n}_i) - \psi_{C}(C_{1:n}) \right\} \\ \times f \left\{\I_{C}(\pobs{U}^{1:n}_j) - \psi_{C}(C_{1:n}) \right\}.
\end{multline*}
On one hand, should the sequence $(\xi_{i,n})_{i \in \Z}$ be an i.i.d.\ multiplier sequence, that is, should it satisfy (M0), unsurprisingly, the above estimator simplifies to
\begin{equation}
\label{eq:checksigmaiid}
\check \sigma_{n,C,f}^2 = \frac{1}{n} \sum_{i=1}^n \left[ f \left\{ \I_{C}(\pobs{U}^{1:n}_i) - \psi_{C}(C_{1:n}) \right\} \right]^2.
\end{equation}
On the other hand, if the multiplier sequence satisfies (M1)--(M3), one obtains 
\begin{multline}
\label{eq:checksigma}
\check \sigma_{n,C,f}^2 = \frac{1}{n} \sum_{i,j=1}^n \varphi \left(\frac{i-j}{\ell_n}\right) f \left\{ \I_{C}(\pobs{U}^{1:n}_i) - \psi_{C}(C_{1:n}) \right\} \\ \times f \left\{\I_{C}(\pobs{U}^{1:n}_j) - \psi_{C}(C_{1:n}) \right\}, 
\end{multline}
which has the form of the HAC kernel estimator of \cite{deJDav00}.

Very naturally, once $C$ has been replaced by~$C_{1:n}$, we use the form in~\eqref{eq:checksigmaiid} (resp.~\eqref{eq:checksigma}) for serially independent (resp.\ weakly dependent) observations. The following result, proved in Section~\ref{proof:prop:convsigma} of the supplementary material, establishes the consistency of $\check \sigma_{n,C_{1:n},f}^2$ under $H_0$.

\begin{prop}
\label{prop:convsigma}
Assume that $f:\R^{2^d-1} \to \R$ in the definition of~\eqref{eq:Snf} is linear and that either
\begin{enumerate}[\bf (i)]
\item the random vectors $\vec X_1,\dots,\vec X_n$ are i.i.d.\ with continuous margins,
\item or the random vectors $\vec X_1,\dots,\vec X_n$ are drawn from a strictly stationary sequence $(\vec X_i)_{i \in \Z}$ with continuous margins whose strong mixing coefficients satisfy $\alpha_r = O(r^{-a})$ for some $a > 6$, and $\ell_n = O(n^{1/2 - \eps})$ for some $0 < \eps < 1/2$ such that, additionally, $\varphi$ defined in~(M3) is twice continuously differentiable on $[-1,1]$ with $\varphi''(0) \neq 0$ and is Lipschitz continuous on $\R$. 
\end{enumerate}
Then, $\check \sigma_{n,C_{1:n},f}^2 \p \sigma_{C,f}^2$. As a consequence, the weak limit of $\check \sigma_{n,C_{1:n},f}^{-1} S_{n,f}$ is equal in distribution to $\sup_{s \in [0,1]} |\U(s)|$.
\end{prop}

As in the previous subsection, better behaved tests are obtained if~\eqref{eq:IbnCA} is used instead of~\eqref{eq:ICA} in the above developments. Let
$$
\I_{b_n,C}(\vec u) = \left( \I_{b_n,C,\{1\}}(\vec u), \I_{b_n,C,\{2\}}(\vec u), \dots, \I_{b_n,C,D}(\vec u) \right), \qquad \vec u \in [0,1]^d,
$$
and let $\tilde \sigma_{n,b_n,C_{1:n},f}^2$ be the corresponding estimator of $\sigma_{C,f}^2$. Proceeding as above, for serially independent data, the appropriate form of $\tilde \sigma_{n,b_n,C_{1:n},f}^2$ is
\begin{equation}
\label{eq:tildesigmaiid}
\tilde \sigma_{n,b_n,C_{1:n},f}^2 = \frac{1}{n} \sum_{i=1}^n \left[ f \left\{ \I_{b_n,C_{1:n}}(\pobs{U}^{1:n}_i) - \bar \I_{b_n,C_{1:n}} \right\} \right]^2,
\end{equation}
where $\bar \I_{b_n,C_{1:n}}  = n^{-1} \sum_{=1}^n \I_{b_n,C_{1:n}}(\pobs{U}^{1:n}_i)$, while, for weakly dependent observations, 
\begin{multline}
\label{eq:tildesigma}
\tilde \sigma_{n,b_n,C_{1:n},f}^2 = \frac{1}{n} \sum_{i,j=1}^n \varphi \left(\frac{i-j}{\ell_n}\right) f \left\{ \I_{b_n,C_{1:n}}(\pobs{U}^{1:n}_i) - \bar \I_{b_n,C_{1:n}} \right\} \\ \times f \left\{\I_{b_n,C_{1:n}}(\pobs{U}^{1:n}_j) - \bar \I_{b_n,C_{1:n}} \right\}.
\end{multline}

The following analogue of Proposition~\ref{prop:convsigma} is proved in Section~\ref{proof:prop:convsigma} of the supplementary material.

\begin{prop}
\label{prop:convsigmatilde}
If $b_n = o(n^{-1/2})$, Proposition~\ref{prop:convsigma} holds with $\check \sigma_{n,C_{1:n},f}^2$ replaced with $\tilde \sigma_{n,b_n,C_{1:n},f}^2$.
\end{prop}

\subsection{Estimation of the bandwidth parameter $\ell_n$}
\label{sec:bandwidth}

When the available observations are weakly dependent, both the approach based on resampling presented in Section~\ref{sec:resampling} and the one based on the estimation of the asymptotic null distribution discussed in Section~\ref{sec:asymp} require the choice of the bandwidth parameter~$\ell_n$. The latter quantity appears in the definition of the dependent multiplier sequences and, as mentioned in \cite{BucKoj14}, plays a role somehow analogous to that of the block length in the block bootstrap. The value of $\ell_n$ is therefore expected to have a crucial influence on the finite-sample performance of the two versions of the test based on $S_{n,f}$ described previously.

The aim of this subsection is to propose an estimator of $\ell_n$ in the spirit of that investigated in \cite{PapPol01}, \cite{PolWhi04} and \cite{PatPolWhi09}, among others, for other resampling schemes. By analogy with~\eqref{eq:checksigmaBn}, we start from the non computable estimator of $\sigma_{C,f}^2$ defined by 
\begin{equation}
\label{eq:sigmaBn}
\sigma_{n,C,f}^2 = \var_\xi [f \circ \psi_C\{\bar \B_n(0,1,\cdot)\}], 
\end{equation}
where
$$
\bar \B_n(s, t,\vec{u}) = \frac{1}{\sqrt{n}} \sum_{i=\ip{ns}+1}^{\ip{nt}} \xi_{i,n} \{ \1 ( \vec U_i \leq \vec u ) - C(\vec u) \}, \qquad (s,t,\vec u) \in \Delta \times [0,1]^d,
$$
and $(\xi_{i,n})_{i \in \Z}$ is a dependent multiplier sequence. Proceeding as for~\eqref{eq:checksigmaBn}, it is easy to verify that
\begin{equation}
\label{eq:sigma_nCf}
\sigma_{n,C,f}^2 = \frac{1}{n} \sum_{i,j=1}^n \varphi \left(\frac{i-j}{\ell_n}\right) f \left\{ \I_C(\vec U_i) - \psi_{C}(C) \right\} f \left\{ \I_C(\vec U_j) - \psi_{C}(C) \right\}. 
\end{equation}
Under the conditions of Proposition~\ref{prop:convsigma}~(ii) and from the fact that the random variables $| f \circ \I_C(\vec U_i) |$ are bounded by $\sup_{\vec x \in [-1,1]^{2^d-1}}|f(\vec x)| < \infty$ (since $\sup_{\vec u \in [0,1]^d} |\I_{C,A}(\vec u)| \leq 1$ for all $A \subseteq D$ $|A| \geq 1$), we can proceed as in the proofs of Propositions 5.1 and 5.2 in \cite{BucKoj14} (see also Lemmas~3.12 and~3.13 in \cite{Buh93} and Proposition 2.1 in \cite{Sha10}) to obtain that
$$
\Ex( \sigma_{n,C,f}^2 ) - \sigma_{C,f}^2 = \frac{\Gamma}{\ell_n^2} + o(\ell_n^{-2}) \qquad \mbox{and} \qquad \var( \sigma_{n,C,f}^2  ) = \frac{\ell_n}{n} \Delta +  o(\ell_n/n),
$$
where $\Gamma = \varphi''(0)/2  \sum_{k=-\infty}^\infty k^2 \tau(k)$ with $\tau(k) = \cov\{ f \circ \I_C(\vec U_0),  f \circ \I_C(\vec U_k) \}$, and $\Delta = 2 \sigma_{C,f}^4 \int_{-1}^1 \varphi(x)^2 \dd x$. As a consequence, the mean squared error of $\sigma_{n,C,f}^2$ is
\begin{equation}
\label{eq:MSE}
\MSE ( \sigma_{n,C,f}^2 ) = \frac{\Gamma^2 }{\ell_n^4} + \Delta \frac{\ell_n}{n} + o(\ell_n^{-4}) + o(\ell_n/n).
\end{equation}
Differentiating the function $x \mapsto \Gamma^2/x^4 + \Delta x/n$ and equating the derivative to zero, we obtain that the value of $\ell_n$ that minimizes the mean square error of $\sigma_{n,C,f}^2$ is, asymptotically, 
$$
\ell_n^{opt} = \left( \frac{4 \Gamma^2 }{\Delta} \right)^{1/5} n^{1/5}.
$$
To estimate $\ell_n^{opt}$, it is necessary to estimate the infinite sum $\sum_{k \in \Z} k^2 \tau(k)$ as well as $\sigma_{C,f}^2 = \sum_{k \in \Z} \tau(k)$ through a {\em pilot} estimate. To do so, we adapt the approach described in \citet[page 1111]{PapPol01} and \citet[Section 3]{PolWhi04} to the current context  \citep[see also][]{PatPolWhi09}. Let $\hat \tau_n(k)$ be the sample autocovariance at lag $k$ computed from the sequence $f \circ \I_{b_n,C_{1:n}}(\pobs{U}^{1:n}_1),\dots,f \circ \I_{b_n,C_{1:n}}(\pobs{U}^{1:n}_n)$. Then, we estimate $\Gamma$ and $\Delta$ by
$$
\hat \Gamma_n = \varphi''(0)/2  \sum_{k=-L}^L \lambda(k/L) k^2 \hat \tau_n(k) 
$$
and
$$
\hat \Delta_n = 2 \left\{ \sum_{k=-L}^L \lambda(k/L) \hat \tau_n(k) \right\}^2 \left\{ \int_{-1}^1 \varphi(x)^2 \dd x \right\},
$$
respectively, where $\lambda(x) = [ \{ 2(1-|x|) \} \vee 0 ] \wedge 1$, $x \in \R$, is the ``flat top'' (trapezoidal) kernel of \cite{PolRom95} and $L$ is an integer estimated by adapting the procedure described in \citet[Section 3.2]{PolWhi04}. Let $\hat \varrho_n(k)$ be the sample autocorrelation at lag $k$ estimated from $f \circ \I_{b_n,C_{1:n}}(\pobs{U}^{1:n}_1),\dots,f \circ \I_{b_n,C_{1:n}}(\pobs{U}^{1:n}_n)$. The parameter $L$ is then taken as the smallest integer $k$ after which $\hat \varrho_n(k)$ appears negligible. The latter is determined automatically by means of the algorithm described in detail in \citet[Section 3.2]{PolWhi04}. Our implementation is based on Matlab code by A.J. Patton (available on his web page) and its \textsf{R} version by J. Racine and C. Parmeter.


\section{Monte Carlo experiments}
\label{sims}

In the previous section, two ways to compute approximate p-values for generic change-point tests based on~\eqref{eq:Snf} were studied under the null. These asymptotic results do not however guarantee that such tests will behave satisfactorily in finite-samples, which is why additional numerical simulations are needed. In our experiments, we restricted attention to the three statistics given in~\eqref{eq:Sni}. For each statistic $S_{n,i}$, $i \in \{1,2,3\}$, an approximate p-value was computed using either the resampling approach based on the processes in~\eqref{eq:tildeSnA}, or the estimated asymptotic null distribution based on variance estimators of the form~\eqref{eq:tildesigmaiid} or~\eqref{eq:tildesigma}. To distinguish between these two situations, we shall talk about {\em the test $\tilde S_{n,i}$} and {\em the test $S_{n,i}^a$}, respectively, in the rest of the paper.

The experiments were carried out in the \textsf{R} statistical system using the \texttt{copula} package \citep{copula}. The sequence $b_n$ involved in both classes of tests was taken equal to $n^{-0.51}$. The only (asymptotically negligible) difference with the theoretical developments presented in the previous sections is that the rescaled maximal ranks in~\eqref{eq:pseudo} were computed by dividing the ranks by $l-k+2$ instead of $l-k+1$. 

\paragraph{Data generating procedure} Two multivariate time series models were used to generate $d$-dimensional samples of size $n$ in our Monte Carlo experiments: a simple autoregressive model of order one and a GARCH(1,1)-like model. Apart from $d$, $n$ and the parameters of the models, the other inputs of the procedure are a real $t \in (0,1)$ determining the location of the possible change-point in the innovations, and two $d$-dimensional copulas $C_1$ and $C_2$. The procedure used to generate a $d$-dimensional sample $\vec X_1,\dots,\vec X_n$ then consists of:
\begin{compactenum}
\item generating independent random vectors $\vec U_i$, $i \in \{-100,\dots,0,\dots,n\}$ such that $\vec U_i$, $i \in \{-100,\dots,0,\dots,\ip{nt}\}$ are i.i.d.\ from copula $C_1$ and $\vec U_i$, $i \in \{\ip{nt}+1,\dots,n\}$ are i.i.d.\ from copula $C_2$,
\item computing $\vec \epsilon_i = (\Phi^{-1}(U_{i1}),\dots,\Phi^{-1}(U_{id}))$, where $\Phi$ is the c.d.f.\ of the standard normal distribution, 
\item setting $\vec X_{-100} = \vec \epsilon_{-100}$ and, for any $j \in D$, computing recursively either
\begin{equation}
\tag{AR1}
\label{eq:ar1}
  X_{ij} = \gamma X_{i-1,j} + \epsilon_{ij}, 
\end{equation} 
or  
\begin{equation}
\tag{GARCH}
\label{eq:garch}
\sigma_{ij}^2 = \omega_j + \beta_j \sigma_{i-1,j}^2 + \alpha_j \epsilon_{i-1,j}^2 \quad \mbox{and} \quad X_{ij} = \sigma_{ij} \epsilon_{ij},
\end{equation}
for $i=-99,\dots,0,\dots,n$.
\end{compactenum}
If the copulas $C_1$ and $C_2$ are chosen equal, the above procedure generates samples under~$H_0$ defined in~\eqref{H0}. Three possible values were considered for the parameter $\gamma$ controlling the strength of the serial dependence in~\eqref{eq:ar1}: 0 (serial independence), 0.25 (mild serial dependence), 0.5 (strong serial dependence). Model~\eqref{eq:garch} was only considered in the bivariate case, and following \cite{BucRup13}, with $(\omega_1,\beta_1,\alpha_1)=(0.012,0.919,0.072)$ and $(\omega_2,\beta_2,\alpha_2)=(0.037,0.868,0.115)$. The latter values were estimated by \cite{JonPooRoc07} from SP500 and DAX daily logreturns, respectively.

Samples under $H_{0,m} \cap (\neg H_{0,c})$, where $H_{0,m}$ and $H_{0,c}$ are defined in~\eqref{H0m} and~\eqref{H0c}, respectively, were obtained by taking $C_1 \neq C_2$ and $t \in \{0.1,0.25,0.5\}$. Notice that when $\gamma = 0$ in~\eqref{eq:ar1}, the latter are samples under $H_{0,m} \cap H_{1,c}$, where
\begin{align*}
\nonumber
  H_{1,c} :\, &
  \exists \text{ distinct } C_1 \text { and } C_2\text{, and } 
  t \in (0,1) \text{ such that }\\
  &\vec X_1, \ldots, \vec X_{\ip{nt}} \text{ have copula } C_1 \text{ and } \vec X_{\ip{nt}+1}, \ldots, \vec X_n \text{ have copula } C_2.
\end{align*}
This is not the case anymore when $\gamma > 0$ as the change in cross-sectional dependence is then gradual by~\eqref{eq:ar1}. 

\paragraph{Other factors of the experiments} Five copula families were considered (the Clayton, the Gumbel--Hougaard, the Normal, the Frank and the Student), the cross-sectional dimensional $d$ was taken in $\{2,4\}$, and the values 50, 100, 200, 400 and 500 were used for~$n$. To estimate the power of the tests, 1000 samples were generated under each combination of factors and all the tests were carried out at the 5\% significance level.

\paragraph{Computation of the test statistics and of the corresponding p-values} The data generating procedure above generates multivariate time series whose component series do not contain ties with probability one. Consequently, as explained in Section~\ref{sec:cpstat}, $S_{n,2}$ is merely $S_{n,1}$ computed from the sample $-\vec X_1,\dots,-\vec X_n$. Furthermore, if $d=2$, it is easy to see that $S_{n,1} = S_{n,2} = S_{n,3}$. However, it can be verified that only the approximate p-values for the tests $\tilde S_{n,1}$ and $\tilde S_{n,3}$ (resp.\ $S_{n,1}^a$ and $S_{n,3}^a$) will be equal. Indeed, the multiplier replicates based on the processes in~\eqref{eq:tildeSnA} (resp.\ the variance estimators of the form~\eqref{eq:tildesigmaiid} or~\eqref{eq:tildesigma})  computed from $\vec X_1,\dots,\vec X_n$ do not coincide in general with those computed from $-\vec X_1,\dots,-\vec X_n$, even in dimension two.

\noindent
From Proposition~\ref{prop:convsigmatilde}, we see that, to compute an asymptotic p-value for the tests $S_{n,i}^a$, it is necessary to be able to compute the c.d.f.\ of the random variable $\sup_{s \in [0,1]} | \U(s) |$. The distribution of the latter random variable is known as the Kolmogorov distribution. As classically done in other contexts, we approach this distribution by that of the statistic of the classical Kolmogorov--Smirnov goodness-of-fit test for a simple hypothesis. Specifically, we use the function {\tt pkolmogorov1x} given in the code of the \textsf{R} function {\tt ks.test}.

\setlength{\tabcolsep}{4pt}

\begin{table}[t!]
\centering
\caption{Percentage of rejection of $H_0$ computed from 1000 samples of size $n \in \{50, 100, 200, 400\}$ generated with $\gamma = 0$ in \eqref{eq:ar1} and when $C_1=C_2=C$ is either the $d$-dimensional Clayton (Cl) or Gumbel--Hougaard (GH) copula the bivariate margins of which have a Kendall's tau of $\tau$. The tests $\tilde S_{n,i}$ are carried out with i.i.d.\ multiplier sequences, while the tests $S_{n,i}^a$ use variance estimators of the form~\eqref{eq:tildesigmaiid}.} 
\label{H0iid}
\begin{tabular}{lrrrrrrrrrrrr}
  \hline
  \multicolumn{3}{c}{} & \multicolumn{4}{c}{$d=2$} & \multicolumn{6}{c}{$d=4$} \\ \cmidrule(lr){4-7} \cmidrule(lr){8-13} $C$ & $n$ & $\tau$ & $\tilde S_{n,1}$ & $\tilde S_{n,2}$ & $S_{n,1}^a$ & $S_{n,2}^a$ & $\tilde S_{n,1}$ & $\tilde S_{n,2}$ & $\tilde S_{n,3}$ & $S_{n,1}^a$ & $S_{n,2}^a$ & $S_{n,3}^a$ \\ \hline
Cl & 50 & 0.1 & 6.8 & 7.4 & 2.6 & 3.0 & 4.6 & 5.1 & 4.0 & 1.2 & 2.1 & 0.7 \\ 
   &  & 0.3 & 4.1 & 5.2 & 1.7 & 4.2 & 4.9 & 5.4 & 3.7 & 0.5 & 2.6 & 0.7 \\ 
   &  & 0.5 & 3.1 & 2.7 & 2.5 & 8.6 & 7.1 & 3.9 & 4.9 & 2.8 & 2.8 & 1.2 \\ 
   &  & 0.7 & 3.0 & 0.5 & 8.3 & 23.8 & 7.4 & 4.1 & 3.3 & 5.4 & 10.3 & 3.1 \\ 
   & 100 & 0.1 & 3.5 & 4.3 & 2.3 & 2.7 & 4.1 & 5.3 & 4.4 & 1.6 & 3.4 & 2.5 \\ 
   &  & 0.3 & 4.0 & 4.4 & 2.3 & 3.6 & 5.7 & 4.7 & 4.4 & 2.0 & 2.8 & 1.4 \\ 
   &  & 0.5 & 4.2 & 4.0 & 4.9 & 8.3 & 4.3 & 4.0 & 3.5 & 2.2 & 3.7 & 1.9 \\ 
   &  & 0.7 & 5.7 & 1.6 & 12.6 & 23.1 & 9.1 & 3.9 & 7.6 & 11.3 & 9.5 & 7.4 \\ 
   & 200 & 0.1 & 4.9 & 4.7 & 2.8 & 3.1 & 6.1 & 5.1 & 5.2 & 3.1 & 3.4 & 3.3 \\ 
   &  & 0.3 & 4.9 & 5.3 & 3.7 & 4.9 & 4.1 & 5.6 & 4.2 & 2.3 & 3.6 & 1.9 \\ 
   &  & 0.5 & 4.6 & 4.3 & 4.8 & 6.9 & 4.6 & 5.5 & 4.2 & 4.1 & 4.8 & 3.2 \\ 
   &  & 0.7 & 5.6 & 3.1 & 11.2 & 15.1 & 10.5 & 5.3 & 11.1 & 14.1 & 8.3 & 9.9 \\ 
   & 400 & 0.1 & 4.6 & 4.9 & 3.7 & 3.8 & 6.3 & 6.7 & 6.5 & 4.5 & 5.5 & 4.8 \\ 
   &  & 0.3 & 4.3 & 4.6 & 4.0 & 4.4 & 5.8 & 5.3 & 5.5 & 4.1 & 4.2 & 3.8 \\ 
   &  & 0.5 & 4.8 & 4.6 & 4.2 & 4.8 & 5.8 & 4.5 & 5.5 & 5.5 & 4.0 & 4.7 \\ 
   &  & 0.7 & 5.9 & 4.0 & 9.3 & 10.8 & 8.5 & 6.6 & 8.7 & 13.5 & 8.1 & 8.2 \\ 
  GH & 50 & 0.1 & 6.7 & 6.3 & 3.4 & 2.3 & 5.8 & 5.3 & 4.7 & 2.4 & 0.8 & 2.5 \\ 
   &  & 0.3 & 4.1 & 3.9 & 3.5 & 2.1 & 5.9 & 6.0 & 5.3 & 1.8 & 0.7 & 3.1 \\ 
   &  & 0.5 & 3.1 & 3.4 & 6.9 & 3.4 & 4.6 & 4.9 & 4.0 & 3.0 & 2.5 & 6.5 \\ 
   &  & 0.7 & 2.0 & 1.8 & 15.5 & 10.7 & 3.4 & 6.2 & 2.0 & 6.2 & 4.2 & 10.3 \\ 
   & 100 & 0.1 & 5.2 & 5.1 & 2.7 & 2.5 & 4.3 & 4.8 & 4.1 & 2.5 & 1.5 & 2.1 \\ 
   &  & 0.3 & 5.9 & 5.3 & 5.2 & 3.9 & 6.1 & 6.7 & 6.7 & 3.1 & 1.9 & 4.5 \\ 
   &  & 0.5 & 3.7 & 3.7 & 6.6 & 5.1 & 5.3 & 4.8 & 5.3 & 3.6 & 3.4 & 6.4 \\ 
   &  & 0.7 & 1.3 & 2.3 & 16.9 & 13.8 & 4.5 & 7.0 & 2.7 & 8.6 & 9.0 & 14.2 \\ 
   & 200 & 0.1 & 5.2 & 5.2 & 3.8 & 3.5 & 4.8 & 4.3 & 4.5 & 3.3 & 2.6 & 3.1 \\ 
   &  & 0.3 & 5.2 & 5.1 & 4.7 & 3.9 & 6.0 & 6.5 & 5.3 & 4.7 & 3.3 & 4.3 \\ 
   &  & 0.5 & 4.5 & 4.5 & 5.2 & 4.7 & 4.2 & 3.9 & 4.0 & 3.2 & 3.6 & 3.9 \\ 
   &  & 0.7 & 2.2 & 3.7 & 12.8 & 10.8 & 4.6 & 7.0 & 4.9 & 6.6 & 9.0 & 10.9 \\ 
   & 400 & 0.1 & 6.4 & 6.1 & 4.8 & 4.7 & 5.1 & 5.7 & 4.3 & 4.0 & 3.1 & 3.1 \\ 
   &  & 0.3 & 4.7 & 4.6 & 4.1 & 3.8 & 4.6 & 5.3 & 5.6 & 3.7 & 3.6 & 4.4 \\ 
   &  & 0.5 & 3.3 & 3.3 & 3.5 & 3.0 & 4.3 & 5.1 & 4.5 & 3.9 & 4.5 & 4.7 \\ 
   &  & 0.7 & 4.6 & 5.8 & 10.1 & 9.9 & 5.3 & 7.1 & 5.9 & 6.3 & 9.5 & 10.4 \\ 
   \hline
\end{tabular}
\end{table}

\begin{table}[t!]
\centering
\caption{Percentage of rejection of $H_0$ computed from 1000 samples of size $n \in \{50, 100, 200\}$ generated with $\gamma=0$ in~\eqref{eq:ar1}, $t \in \{0.1,0.25,0.5\}$ and when $C_1$ and $C_2$ are both $d$-dimensional normal (N) or Frank (F) copulas such that the bivariate margins of $C_1$ have a Kendall's tau of 0.2 and those of $C_2$ a Kendall's tau of $\tau$. The colunms CvM give the results for the test studied in \cite{BucKojRohSeg14}. All the tests were carried out with i.i.d.\ multiplier sequences.} 
\label{H1iid}
\begin{tabular}{lrrrrrrrrrr}
  \hline
  \multicolumn{4}{c}{} & \multicolumn{3}{c}{$d=2$} & \multicolumn{4}{c}{$d=4$} \\ \cmidrule(lr){5-7} \cmidrule(lr){8-11} $C$ & $n$ & $\tau$ & $t$  & CvM & $\tilde S_{n,1}$ & $\tilde S_{n,2}$ & CvM & $\tilde S_{n,1}$ & $\tilde S_{n,2}$ & $\tilde S_{n,3}$ \\ \hline
N & 50 & 0.4 & 0.10 & 5.6 & 6.0 & 5.6 & 5.9 & 7.9 & 7.9 & 8.3 \\ 
   &  &  & 0.25 & 9.1 & 8.7 & 8.9 & 12.2 & 17.3 & 18.9 & 19.5 \\ 
   &  &  & 0.50 & 13.4 & 12.6 & 12.6 & 24.3 & 25.1 & 27.6 & 28.2 \\ 
   &  & 0.6 & 0.10 & 9.0 & 8.7 & 8.9 & 7.1 & 20.7 & 21.7 & 22.4 \\ 
   &  &  & 0.25 & 32.3 & 34.7 & 32.6 & 45.6 & 66.3 & 67.0 & 69.9 \\ 
   &  &  & 0.50 & 46.7 & 42.7 & 41.6 & 76.1 & 78.0 & 77.5 & 80.8 \\ 
   & 100 & 0.4 & 0.10 & 5.7 & 7.8 & 7.6 & 7.6 & 11.2 & 12.2 & 12.3 \\ 
   &  &  & 0.25 & 14.9 & 19.7 & 19.1 & 27.0 & 35.3 & 37.2 & 43.0 \\ 
   &  &  & 0.50 & 25.9 & 28.9 & 29.2 & 54.5 & 54.6 & 53.5 & 59.6 \\ 
   &  & 0.6 & 0.10 & 14.6 & 22.7 & 23.4 & 26.1 & 47.5 & 51.1 & 58.8 \\ 
   &  &  & 0.25 & 60.0 & 68.6 & 69.0 & 90.3 & 94.9 & 94.8 & 97.6 \\ 
   &  &  & 0.50 & 81.9 & 84.8 & 84.2 & 98.8 & 98.4 & 99.0 & 99.5 \\ 
   & 200 & 0.4 & 0.10 & 9.1 & 11.7 & 12.3 & 13.2 & 18.2 & 17.9 & 23.3 \\ 
   &  &  & 0.25 & 26.5 & 36.7 & 36.9 & 58.9 & 64.9 & 67.1 & 75.5 \\ 
   &  &  & 0.50 & 47.7 & 54.2 & 53.7 & 83.4 & 83.5 & 83.3 & 88.9 \\ 
   &  & 0.6 & 0.10 & 34.5 & 57.7 & 58.0 & 63.1 & 87.3 & 87.8 & 93.8 \\ 
   &  &  & 0.25 & 92.6 & 96.5 & 96.7 & 100.0 & 100.0 & 100.0 & 100.0 \\ 
   &  &  & 0.50 & 99.1 & 99.5 & 99.5 & 100.0 & 100.0 & 100.0 & 100.0 \\ 
  F & 50 & 0.4 & 0.10 & 6.9 & 5.7 & 6.2 & 4.5 & 7.8 & 9.0 & 8.4 \\ 
   &  &  & 0.25 & 10.8 & 9.7 & 10.0 & 12.9 & 17.9 & 19.7 & 19.9 \\ 
   &  &  & 0.50 & 15.1 & 13.6 & 13.6 & 24.7 & 30.2 & 31.1 & 29.1 \\ 
   &  & 0.6 & 0.10 & 11.1 & 10.6 & 11.3 & 7.3 & 23.3 & 29.7 & 24.8 \\ 
   &  &  & 0.25 & 33.1 & 32.7 & 31.9 & 42.3 & 67.2 & 70.2 & 69.5 \\ 
   &  &  & 0.50 & 50.9 & 46.1 & 46.2 & 78.3 & 81.9 & 82.3 & 85.5 \\ 
   & 100 & 0.4 & 0.10 & 6.1 & 7.0 & 7.4 & 6.5 & 9.2 & 13.6 & 11.9 \\ 
   &  &  & 0.25 & 16.5 & 18.2 & 18.7 & 26.5 & 38.8 & 46.8 & 49.6 \\ 
   &  &  & 0.50 & 26.4 & 28.6 & 28.3 & 48.9 & 52.7 & 58.3 & 61.6 \\ 
   &  & 0.6 & 0.10 & 17.7 & 27.3 & 27.2 & 22.7 & 55.3 & 63.9 & 68.6 \\ 
   &  &  & 0.25 & 66.5 & 73.6 & 74.0 & 91.9 & 97.7 & 98.2 & 99.5 \\ 
   &  &  & 0.50 & 86.2 & 87.3 & 87.5 & 99.3 & 98.8 & 99.4 & 99.8 \\ 
   & 200 & 0.4 & 0.10 & 10.2 & 15.7 & 15.6 & 12.5 & 19.7 & 25.3 & 27.1 \\ 
   &  &  & 0.25 & 34.3 & 41.3 & 41.5 & 53.6 & 64.4 & 76.2 & 78.8 \\ 
   &  &  & 0.50 & 50.7 & 54.3 & 54.4 & 83.2 & 83.9 & 90.4 & 93.2 \\ 
   &  & 0.6 & 0.10 & 39.0 & 64.7 & 65.6 & 60.3 & 88.0 & 92.2 & 96.4 \\ 
   &  &  & 0.25 & 95.4 & 98.3 & 98.3 & 99.9 & 100.0 & 100.0 & 100.0 \\ 
   &  &  & 0.50 & 99.5 & 99.8 & 99.8 & 100.0 & 100.0 & 100.0 & 100.0 \\ 
   \hline
\end{tabular}
\end{table}

\paragraph{Empirical levels and power of the tests based on i.i.d.\ multipliers / a variance estimator of the form~\eqref{eq:tildesigmaiid}} Table~\ref{H0iid} gives the empirical levels of the tests when the observations are serially independent. For the sake of brevity, the results are reported only for two copula families. Overall, we find that the tests $\tilde S_{n,i}$ with multiplier sequences satisfying (M0) (here standard normal sequences) hold there level rather well both for $d=2$ and $d=4$, and all the considered degrees of cross-sectional dependence. This is not the case for the tests $S_{n,i}^a$ which frequently appear way too liberal when the cross-sectional dependence is high.

\noindent
Table~\ref{H1iid} partially reports the percentages of rejection of the i.i.d.\ multiplier tests for serially independent observations generated under $H_{0,m} \cap H_{1,c}$ resulting from a change of the copula parameter within a copula family. The columns CvM give the results of the i.i.d.\ multiplier test based on the maximally selected Cram\'er--von Mises statistic studied in~\citet{BucKojRohSeg14} (with multiplier replicates of the form~(4.6) in the latter reference) and implemented in the \textsf{R} package {\tt npcp}. Overall, we find that the tests $\tilde S_{n,i}$ are more powerful than that studied in~\cite{BucKojRohSeg14} for such scenarios, especially when the change in the copula occurs early or late. Among the tests $\tilde S_{n,i}$, we observed that the test $\tilde S_{n,3}$ (which coincides with the test $\tilde S_{n,1}$ in dimension two) led frequently to slightly higher rejection rates, although this conclusion is based on a limited number of simulation scenarios. The rejection rates of the tests $S_{n,i}^a$ with a variance estimator of the form~\eqref{eq:tildesigmaiid} are not reported for the sake of brevity. They were found to be slightly less powerful than the tests $\tilde S_{n,i}$ when $\tau = 0.4$. For $\tau=0.6$, a comparison of the two classes of tests is not necessarily meaningful as the tests $S_{n,i}^a$ were often found to be way too liberal under strong cross-sectional dependence.

\begin{table}[t!]
\centering
\caption{Percentage of rejection of $H_0$ computed from 1000 samples of size $n \in \{100, 200, 400\}$ when $C_1 = C_2 = C$ is either the bivariate Clayton (Cl), Gumbel--Hougaard (GH) or Frank (F) copula with a Kendall's tau of $\tau$. In the first four vertical blocks of the table, the test $\tilde S_{n,1}$ (resp.\ $S_{n,1}^a$) is carried out using dependent multiplier sequences (resp.\ a variance estimator of the form~\eqref{eq:tildesigma}). In the last vertical block, i.i.d.\ multipliers and a variance estimator of the form~\eqref{eq:tildesigmaiid} are used instead.} 
\label{H0sm}
\begin{tabular}{lrrrrrrrrrrrr}
  \hline
  \multicolumn{3}{c}{} & \multicolumn{2}{c}{$\gamma=0$} & \multicolumn{2}{c}{$\gamma=0.25$} & \multicolumn{2}{c}{$\gamma=0.5$} & \multicolumn{2}{c}{GARCH} & \multicolumn{2}{c}{$\gamma=0.5$/ind} \\ \cmidrule(lr){4-5} \cmidrule(lr){6-7} \cmidrule(lr){8-9} \cmidrule(lr){10-11} \cmidrule(lr){12-13} $C$ & $n$ & $\tau$ & $\tilde S_{n,1}$ & $S_{n,1}^a$ & $\tilde S_{n,1}$ & $S_{n,1}^a$ & $\tilde S_{n,1}$ & $S_{n,1}^a$ & $\tilde S_{n,1}$ & $S_{n,1}^a$ & $\tilde S_{n,1}$ & $S_{n,1}^a$ \\ \hline
Cl & 100 & 0.10 & 5.2 & 2.3 & 6.6 & 3.5 & 8.2 & 3.3 & 6.2 & 2.5 & 14.5 & 10.2 \\ 
   &  & 0.30 & 3.5 & 1.8 & 6.7 & 3.1 & 7.1 & 4.7 & 5.2 & 3.3 & 15.0 & 11.6 \\ 
   &  & 0.50 & 4.0 & 3.4 & 5.0 & 4.5 & 5.2 & 4.7 & 4.6 & 4.5 & 12.0 & 13.5 \\ 
   &  & 0.70 & 8.3 & 12.0 & 7.5 & 11.8 & 7.2 & 11.2 & 7.2 & 13.2 & 8.9 & 20.0 \\ 
   & 200 & 0.10 & 4.2 & 2.3 & 5.1 & 2.8 & 6.9 & 3.6 & 5.0 & 3.1 & 17.2 & 13.5 \\ 
   &  & 0.30 & 5.1 & 2.6 & 6.2 & 3.4 & 7.2 & 4.4 & 5.3 & 3.8 & 15.7 & 13.0 \\ 
   &  & 0.50 & 4.4 & 4.1 & 5.0 & 5.1 & 4.6 & 5.1 & 4.5 & 4.5 & 14.1 & 14.2 \\ 
   &  & 0.70 & 6.5 & 12.2 & 6.6 & 9.8 & 7.4 & 11.2 & 6.5 & 10.8 & 12.4 & 20.0 \\ 
   & 400 & 0.10 & 4.7 & 3.3 & 5.6 & 4.3 & 6.0 & 3.5 & 5.3 & 3.8 & 19.4 & 16.9 \\ 
   &  & 0.30 & 4.4 & 3.4 & 6.3 & 4.3 & 6.0 & 4.2 & 4.0 & 3.5 & 17.3 & 15.2 \\ 
   &  & 0.50 & 4.7 & 4.7 & 5.9 & 5.7 & 5.6 & 5.0 & 6.1 & 5.7 & 14.6 & 14.2 \\ 
   &  & 0.70 & 6.4 & 8.7 & 5.7 & 7.9 & 5.1 & 6.8 & 6.6 & 9.5 & 15.7 & 19.0 \\ 
  GH & 100 & 0.10 & 4.8 & 2.5 & 5.1 & 2.0 & 7.7 & 2.7 & 5.6 & 2.8 & 15.3 & 11.2 \\ 
   &  & 0.30 & 5.0 & 3.7 & 5.9 & 4.4 & 7.5 & 4.5 & 4.9 & 2.9 & 15.0 & 14.2 \\ 
   &  & 0.50 & 4.5 & 6.7 & 4.3 & 7.1 & 6.3 & 7.9 & 4.9 & 6.9 & 10.7 & 15.7 \\ 
   &  & 0.70 & 3.5 & 16.0 & 4.3 & 18.9 & 5.1 & 18.9 & 3.7 & 16.2 & 4.5 & 25.4 \\ 
   & 200 & 0.10 & 6.4 & 3.9 & 5.6 & 3.7 & 7.3 & 3.9 & 5.8 & 3.8 & 18.2 & 14.1 \\ 
   &  & 0.30 & 6.0 & 5.1 & 6.4 & 4.6 & 6.7 & 4.6 & 5.4 & 4.5 & 19.1 & 16.4 \\ 
   &  & 0.50 & 5.1 & 4.9 & 6.0 & 6.4 & 6.9 & 8.0 & 3.7 & 4.9 & 15.6 & 17.2 \\ 
   &  & 0.70 & 3.8 & 14.4 & 2.8 & 13.0 & 4.4 & 12.4 & 3.5 & 12.2 & 10.0 & 25.4 \\ 
   & 400 & 0.10 & 5.0 & 4.0 & 5.8 & 4.8 & 6.3 & 5.1 & 5.2 & 3.9 & 18.5 & 16.3 \\ 
   &  & 0.30 & 4.1 & 3.0 & 5.1 & 4.3 & 6.3 & 4.6 & 4.9 & 4.1 & 18.5 & 17.2 \\ 
   &  & 0.50 & 3.2 & 3.6 & 5.0 & 6.3 & 7.9 & 7.5 & 4.9 & 4.7 & 16.7 & 17.2 \\ 
   &  & 0.70 & 5.2 & 9.8 & 3.8 & 8.7 & 5.4 & 10.6 & 3.8 & 8.2 & 14.5 & 22.4 \\ 
  F & 100 & 0.10 & 5.5 & 2.1 & 5.3 & 2.3 & 10.6 & 4.2 & 5.0 & 2.4 & 15.2 & 10.2 \\ 
   &  & 0.30 & 4.4 & 2.2 & 5.9 & 3.9 & 7.7 & 4.1 & 6.4 & 4.7 & 13.3 & 10.3 \\ 
   &  & 0.50 & 4.0 & 7.6 & 4.0 & 6.0 & 5.4 & 7.1 & 4.2 & 6.7 & 12.8 & 18.0 \\ 
   &  & 0.70 & 5.2 & 29.3 & 4.8 & 26.5 & 5.4 & 18.1 & 5.4 & 23.9 & 5.9 & 28.5 \\ 
   & 200 & 0.10 & 4.0 & 2.1 & 6.0 & 3.9 & 8.3 & 4.5 & 5.1 & 2.9 & 17.5 & 13.4 \\ 
   &  & 0.30 & 5.0 & 3.9 & 5.7 & 4.1 & 7.1 & 3.9 & 5.3 & 3.4 & 17.0 & 14.5 \\ 
   &  & 0.50 & 4.8 & 6.2 & 4.5 & 5.7 & 6.9 & 7.1 & 4.4 & 5.6 & 15.0 & 17.3 \\ 
   &  & 0.70 & 3.2 & 19.9 & 4.0 & 17.5 & 4.6 & 13.4 & 4.9 & 20.1 & 8.9 & 25.1 \\ 
   & 400 & 0.10 & 4.1 & 3.1 & 6.0 & 4.4 & 6.0 & 4.0 & 4.5 & 3.0 & 18.0 & 14.8 \\ 
   &  & 0.30 & 5.5 & 4.6 & 6.7 & 5.6 & 5.9 & 4.2 & 5.2 & 4.3 & 14.7 & 12.5 \\ 
   &  & 0.50 & 4.6 & 4.7 & 4.7 & 5.0 & 4.0 & 3.8 & 4.8 & 5.5 & 15.7 & 16.5 \\ 
   &  & 0.70 & 5.3 & 13.2 & 4.5 & 12.3 & 6.2 & 9.9 & 5.7 & 13.2 & 14.2 & 21.7 \\ 
   \hline
\end{tabular}
\end{table}

\begin{table}[t!]
\centering
\caption{Percentage of rejection of $H_0$ computed from 1000 samples of size $n \in \{100, 200\}$ generated with $t \in \{0.1,0.25,0.5\}$ and when $C_1$ and $C_2$ are both bivariate Clayton (Cl), Gumbel--Hougaard (GH) or normal (N) copulas with a Kendall's tau of 0.2 for $C_1$ and a Kendall's tau of $\tau$ for $C_2$. The colunms CvM give the results for the test studied in \cite{BucKojRohSeg14}. The latter test and the test $\tilde S_{n,1}$ (resp.\ the test $S_{n,1}^a$) are (resp.\ is) carried out using dependent multiplier sequences (resp.\ a variance estimator of the form~\eqref{eq:tildesigma}).} 
\label{H1sm}
\begin{tabular}{lrrrrrrrrrrrr}
  \hline
  \multicolumn{4}{c}{} & \multicolumn{3}{c}{$\gamma=0$} & \multicolumn{3}{c}{$\gamma=0.5$} & \multicolumn{3}{c}{GARCH}  \\ \cmidrule(lr){5-7} \cmidrule(lr){8-10} \cmidrule(lr){11-13} $C$ & $n$ & $\tau$ & $t$ & CvM & $\tilde S_{n,1}$ & $S_{n,1}^a$ & CvM & $\tilde S_{n,1}$ & $S_{n,1}^a$ & CvM & $\tilde S_{n,1}$ & $S_{n,1}^a$ \\ \hline
Cl & 100 & 0.4 & 0.10 & 6.5 & 6.5 & 4.3 & 6.5 & 8.0 & 5.0 & 6.6 & 6.7 & 3.8 \\ 
   &  &  & 0.25 & 17.9 & 20.4 & 13.4 & 14.0 & 19.7 & 10.6 & 17.2 & 18.1 & 11.2 \\ 
   &  &  & 0.50 & 23.5 & 23.2 & 15.0 & 18.3 & 22.4 & 9.7 & 28.6 & 27.6 & 17.1 \\ 
   &  & 0.6 & 0.10 & 12.6 & 20.6 & 19.7 & 9.4 & 17.1 & 17.0 & 13.9 & 20.1 & 19.4 \\ 
   &  &  & 0.25 & 61.3 & 65.7 & 52.7 & 44.2 & 53.6 & 36.4 & 61.1 & 64.8 & 50.7 \\ 
   &  &  & 0.50 & 80.0 & 78.8 & 61.1 & 58.4 & 61.8 & 34.9 & 80.3 & 78.3 & 59.3 \\ 
   & 200 & 0.4 & 0.10 & 8.2 & 9.6 & 7.5 & 6.9 & 10.4 & 7.0 & 8.3 & 11.1 & 8.9 \\ 
   &  &  & 0.25 & 26.5 & 31.8 & 25.2 & 19.9 & 27.7 & 20.2 & 27.8 & 32.0 & 26.2 \\ 
   &  &  & 0.50 & 45.3 & 47.0 & 37.0 & 34.2 & 40.0 & 27.9 & 47.1 & 48.8 & 40.1 \\ 
   &  & 0.6 & 0.10 & 30.4 & 42.1 & 42.3 & 12.6 & 28.8 & 28.6 & 29.7 & 43.9 & 43.4 \\ 
   &  &  & 0.25 & 93.2 & 94.2 & 87.4 & 71.1 & 79.2 & 65.9 & 91.1 & 92.2 & 83.5 \\ 
   &  &  & 0.50 & 98.5 & 98.3 & 94.1 & 89.5 & 90.5 & 80.1 & 98.7 & 98.2 & 94.1 \\ 
  GH & 100 & 0.4 & 0.10 & 5.3 & 8.0 & 7.1 & 5.0 & 8.2 & 7.1 & 6.3 & 7.6 & 6.9 \\ 
   &  &  & 0.25 & 12.4 & 17.1 & 12.1 & 11.6 & 18.6 & 11.1 & 14.9 & 18.6 & 14.9 \\ 
   &  &  & 0.50 & 22.5 & 25.2 & 16.9 & 18.2 & 24.2 & 14.0 & 26.0 & 27.7 & 19.9 \\ 
   &  & 0.6 & 0.10 & 10.4 & 18.5 & 26.1 & 7.7 & 19.4 & 25.7 & 10.2 & 19.9 & 26.6 \\ 
   &  &  & 0.25 & 53.3 & 63.1 & 54.7 & 41.2 & 58.0 & 43.7 & 55.0 & 63.8 & 52.4 \\ 
   &  &  & 0.50 & 78.1 & 80.4 & 67.4 & 62.7 & 69.5 & 46.1 & 76.0 & 76.3 & 63.1 \\ 
   & 200 & 0.4 & 0.10 & 7.0 & 10.5 & 10.0 & 7.1 & 11.4 & 9.9 & 6.9 & 10.2 & 9.0 \\ 
   &  &  & 0.25 & 25.2 & 31.9 & 27.7 & 19.1 & 30.9 & 22.8 & 24.6 & 32.3 & 26.7 \\ 
   &  &  & 0.50 & 43.0 & 48.3 & 42.1 & 31.4 & 39.3 & 30.0 & 43.2 & 49.1 & 41.3 \\ 
   &  & 0.6 & 0.10 & 25.9 & 42.7 & 47.2 & 13.0 & 30.1 & 34.0 & 23.5 & 43.4 & 46.3 \\ 
   &  &  & 0.25 & 89.0 & 92.9 & 86.3 & 72.1 & 83.5 & 70.0 & 88.9 & 94.5 & 85.0 \\ 
   &  &  & 0.50 & 98.3 & 98.5 & 95.9 & 89.6 & 92.0 & 83.4 & 98.4 & 98.7 & 93.6 \\ 
  N & 100 & 0.4 & 0.10 & 6.1 & 7.8 & 6.2 & 6.9 & 10.2 & 7.8 & 6.1 & 7.0 & 5.5 \\ 
   &  &  & 0.25 & 14.4 & 19.3 & 14.7 & 13.7 & 19.2 & 13.2 & 14.7 & 17.8 & 13.3 \\ 
   &  &  & 0.50 & 25.6 & 27.7 & 19.4 & 17.5 & 24.1 & 12.5 & 25.2 & 28.7 & 19.2 \\ 
   &  & 0.6 & 0.10 & 10.6 & 27.1 & 32.0 & 8.2 & 19.7 & 23.7 & 10.2 & 19.3 & 24.7 \\ 
   &  &  & 0.25 & 61.5 & 70.1 & 61.3 & 46.0 & 62.3 & 44.8 & 58.4 & 69.2 & 59.3 \\ 
   &  &  & 0.50 & 82.6 & 85.1 & 72.3 & 64.9 & 71.3 & 44.9 & 79.0 & 82.0 & 65.7 \\ 
   & 200 & 0.4 & 0.10 & 8.0 & 10.8 & 9.2 & 5.9 & 12.6 & 9.2 & 7.0 & 9.3 & 8.9 \\ 
   &  &  & 0.25 & 27.7 & 37.4 & 33.2 & 20.4 & 31.0 & 24.7 & 26.8 & 35.1 & 30.7 \\ 
   &  &  & 0.50 & 47.0 & 51.5 & 43.6 & 33.2 & 41.7 & 30.7 & 43.0 & 49.5 & 41.3 \\ 
   &  & 0.6 & 0.10 & 27.1 & 47.3 & 49.6 & 14.5 & 35.6 & 39.2 & 28.8 & 48.3 & 51.8 \\ 
   &  &  & 0.25 & 91.5 & 96.5 & 88.4 & 72.3 & 85.2 & 71.0 & 90.7 & 96.1 & 85.7 \\ 
   &  &  & 0.50 & 98.8 & 99.7 & 96.3 & 91.7 & 95.5 & 83.6 & 99.1 & 99.3 & 94.8 \\ 
   \hline
\end{tabular}
\end{table}

\begin{table}[t!]
\centering
\caption{Percentage of rejection of $H_0$ computed from 1000 samples of size $n=500$ generated with $\gamma = 0$ in~\eqref{eq:ar1} and when $C_1$ and $C_2$ are both either bivariate Student copulas with 1 d.f.\ ($t_1$), with 3 d.f.\ ($t_3$) or with 5 d.f.\ ($t_5$) with a Spearman's rho of 0.4 for $C_1$ and a Spearman's rho of $\rho$ for $C_2$. The test $\tilde S_{n,1}$ was carried out with dependent multiplier sequences, while the test $S_{n,1}^a$ used a variance estimator of the form~\eqref{eq:tildesigma}. The columns W contain the rejection rates of the similar test studied in \cite{WieDehvanVog14}. The results are taken from Table~1 in the latter reference.} 
\label{Wied}
\begin{tabular}{rrrrrrrrrr}
  \hline
  \multicolumn{1}{c}{} & \multicolumn{3}{c}{$t_1$} & \multicolumn{3}{c}{$t_3$} & \multicolumn{3}{c}{$t_5$} \\ \cmidrule(lr){2-4} \cmidrule(lr){5-7} \cmidrule(lr){8-10} $\rho$ & W & $\tilde S_{n,1}$ & $S_{n,1}^a$ & W & $\tilde S_{n,1}$ & $S_{n,1}^a$ & W & $\tilde S_{n,1}$ & $S_{n,1}^a$ \\ \hline
0.4 & 4.5 & 3.9 & 2.8 & 4.5 & 5.2 & 4.0 & 4.7 & 6.3 & 4.4 \\ 
  0.6 & 8.1 & 43.3 & 38.7 & 8.5 & 57.9 & 54.3 & 8.5 & 66.5 & 63.8 \\ 
  0.8 & 20.5 & 99.4 & 98.6 & 21.7 & 100.0 & 99.9 & 21.5 & 100.0 & 100.0 \\ 
  0.2 & 7.9 & 33.7 & 29.2 & 8.8 & 51.0 & 46.6 & 8.9 & 52.9 & 48.4 \\ 
  0.0 & 19.9 & 87.7 & 84.7 & 23.0 & 95.7 & 94.9 & 24.0 & 97.2 & 96.3 \\ 
  -0.2 & 41.8 & 99.7 & 99.6 & 49.5 & 100.0 & 100.0 & 51.5 & 100.0 & 100.0 \\ 
  -0.4 & 70.2 & 100.0 & 100.0 & 78.6 & 100.0 & 100.0 & 80.4 & 100.0 & 99.9 \\ 
  -0.6 & 91.7 & 100.0 & 99.9 & 95.8 & 100.0 & 100.0 & 96.6 & 100.0 & 100.0 \\ 
   \hline
\end{tabular}
\end{table}

\paragraph{Empirical levels and power of the tests based on dependent multipliers / a variance estimator of the form~\eqref{eq:tildesigma}} Part of Table~\ref{H0sm} reports the empirical levels of the test $\tilde S_{n,1}$ when dependent multiplier sequences satisfying (M1)--(M3) are used. These sequences were generated using the ``moving average approach'' proposed initially in \citet[Section~6.2]{Buh93} and revisited in \citet[Section~5.2]{BucKoj14}. A standard normal sequence was used for the required initial i.i.d.\ sequence. The kernel function $\kappa$ in that approach was chosen to be the Parzen kernel defined by $\kappa_{P}(x) = (1 - 6x^2 + 6|x|^3) \1(|x| \leq 1/2) + 2(1-|x|)^3\1(1/2 < |x| \leq 1)$, $x \in \R$, which amounts to choosing the function $\varphi$ in~(M3) as $x \mapsto (\kappa_P \star \kappa_P)(2x) / (\kappa_P \star \kappa_P)(0)$, where `$\star$' denotes the convolution operator. The value of the bandwidth parameter $\ell_n$ defined in~(M2) was estimated using the data-driven procedure described in Section~\ref{sec:bandwidth}. The same value of $\ell_n$ was used to carry out the test $S_{n,1}^a$ relying on a variance estimator of the form~\eqref{eq:tildesigma}.

\noindent
From the first three vertical blocks of Table~\ref{H0sm}, we see that an increase in the degree of serial dependence in~\eqref{eq:ar1} (controlled by $\gamma$) appears to result in a small inflation of the empirical levels of the test $\tilde S_{n,1}$. As expected, the situation improves as $n$ increases from 100 to~400. For sequences generated using~\eqref{eq:garch}, the empirical levels of the test $\tilde S_{n,1}$ appear always reasonably close to the 5\% nominal level. The test $S_{n,1}^a$ remains overall way too liberal when the cross-sectional dependence is high. 

\noindent 
The last vertical block of Table~\ref{H0sm} reports, for strongly serially dependent observations generated using~\eqref{eq:ar1}, the empirical levels of the test $\tilde S_{n,1}$ based on i.i.d.\ multipliers, as well as those of the test $S_{n,1}^a$ based on an inappropriate variance estimator of the form~\eqref{eq:tildesigmaiid}. As expected, both tests strongly fail to hold their level.

\noindent
Table~\ref{H1sm} partially reports the rejection percentages of the tests based on dependent multipliers / a variance estimator of the form~\eqref{eq:tildesigma} for observations generated under $H_{0,m} \cap (\neg H_{0,c})$ resulting from a change of the copula parameter within a copula family. The rejection rates of the test $S_{n,1}^a$ should be considered with care when $\tau=0.6$ as that test was found to be way too liberal under strong cross-sectional dependence. Despite that issue, the test $\tilde S_{n,1}$ appears almost always more powerful than the test $S_{n,1}^a$. Also, as it could have been expected, the presence of strong serial dependence ($\gamma=0.5$) leads to lower rejection percentages when compared with serial independence ($\gamma=0$). Finally, comparing the results for the test $\tilde S_{n,1}$ when $\gamma=0$ with the analogue results reported in Table~\ref{H1iid} reveals that, rather naturally, the use of dependent multipliers in the case of serially independent observations results in a small loss of power.

\noindent
We end this section by a comparison of the tests $\tilde S_{n,1}$ and $S_{n,1}^a$ with the similar test studied in \cite{WieDehvanVog14}. To do so, we reproduced one of the experiments carried out in the latter reference. The results are reported in Table~\ref{Wied} and confirm that tests for change-point detection based on~\eqref{eq:Ckl} are potentially substantially more powerful than tests based on~\eqref{eq:Ckln}.


\section{Practical recommendations and illustration}

Based on the experiments partially reported in the previous section, we recommend, among the tests $\tilde S_{n,i}$ and $S_{n,i}^a$, the tests $\tilde S_{n,i}$. Indeed, the tests $S_{n,i}^a$ did not hold their level well in the case of strong cross-sectional dependence. Furthermore, because of their form, the tests $S_{n,i}^a$ might suffer from some of the practical issues described in \cite{ShaZha10}, and, in future research, it might be of interest to study a {\em self-normalization} version of these as advocated in the latter reference. 

The pros and cons of the tests $\tilde S_{n,i}$ compared with the test studied in \cite{BucKojRohSeg14} are as follows. The tests $\tilde S_{n,i}$ seem more powerful for alternatives involving a change in Spearman's rho at constant margins; they are also substantially faster to compute. Their main weakness is that, by construction, they have no power against alternatives involving a change in the copula at a constant value of Spearman's rho and constant margins.

Among the tests $\tilde S_{n,i}$, we recommend the test $\tilde S_{n,3}$, merely because of its slightly better finite-sample behavior in our simulations.
 
We end this section by a brief illustration of the studied tests on real financial observations. Specifically, we consider a trivariate version of the data analyzed in \citet[Section~7]{DehVogWenWie14}. The observations consist of $n=990$ daily logreturns computed from the DAX, the CAC 40 and the Standard and Poor 500 indices for the years 2006--2009. 
An approximate p-value of 0.045 was obtained for the test $\tilde S_{n,3}$ with dependent multipliers, providing some evidence against $H_0$. It is however important to bear in mind that it is only under the assumption that $H_{0,m}$ in~\eqref{H0m} holds that it would be fully justified to decide to reject $H_{0,c}$ in~\eqref{H0c}.

\section{Conclusion}

Tests for change-point detection based on the generic statistic $S_{n,f}$ defined in~\eqref{eq:Snf} were first studied theoretically. These tests, designed to be particularly sensitive to changes in the cross-sectional dependence of multivariate time series, can be carried out using either resampling based on multipliers, or by estimating the asymptotic null distribution of $S_{n,f}$. Both approaches were shown to be asymptotically valid under strong mixing and suitable conditions on the underlying function $f$. In addition, a procedure for estimating a key bandwidth parameter involved in both techniques for computing p-values was suggested, making the tests fully data-driven. Next, their finite-sample behavior was investigated by means of extensive simulations for three particular choices of the function $f$ resulting in the test statistics defined in~\eqref{eq:Sni} measuring changes in the cross-sectional dependence in terms of multivariate extensions of Spearman's rho. Practical recommendations and an illustration were finally given.

\section*{Acknowledgements}
The authors are grateful to Axel B\"ucher and Johan Segers for fruitful discussions on related projects that led to improvements in this one.

\bibliographystyle{plainnat}
\bibliography{biblio}

\appendix

\section{Proof of Proposition~\ref{prop:weak_SnA_sm}}
\label{proof:prop:weak_SnA_sm}

Let us first introduce some additional notation. For integers $1 \le k \le l \leq n$, let $H_{k:l}$ denote the empirical c.d.f.\ of the unobservable sample $\vec U_{k}, \dots, \vec U_{l}$ and let $H_{k:l,1},\dots,H_{k:l,d}$ denote its margins. The corresponding empirical quantile functions are
\[
  H_{k:l,j}^{-1}(u) = \inf \{ v \in [0, 1] : H_{k:l,j}(v) \ge u \}, \qquad u \in [0, 1], j \in D.
\]
Finally, for any $\vec{u} \in [0, 1]^d$, let
\begin{equation}
\label{eq:hnkl}
  \vec{h}_{k:l}(\vec{u}) = 
  \bigl( H_{k:l,1}(u_1), \ldots, H_{k:l,d}(u_d) \bigr)
\end{equation}
and
\begin{equation}
\label{eq:hnklinv}
  \vec{h}_{k:l}^{-1}(\vec{u}) = 
  \bigl( H_{k:l,1}^{-1}(u_1), \ldots, H_{k:l,d}^{-1}(u_d) \bigr).
\end{equation}
By convention, all the quantities defined above are taken equal to zero if $k > l$.

\begin{proof}[\bf Proof of Proposition~\ref{prop:weak_SnA_sm}]
Fix $A \subseteq D$, $|A| \geq 1$, and $(s,t) \in \Delta$ such that $\ip{ns} < \ip{nt}$. On one hand, from~\eqref{eq:SnA} and by linearity of $\phi_A$ defined in~\eqref{eq:phiA}, we have
$$
\S_{n,A}(s,t) = \frac{1}{\sqrt{n}} \sum_{i=\ip{ns}+1}^{\ip{nt}}  \prod_{j \in A} \{ 1 - H_{\ip{ns}+1:\ip{nt},j}(U_{ij}) \} -  \sqrt{n} \lambda_n(s,t) \phi_A(C), 
$$
where we have used the fact that $\hat U_{ij}^{\ip{ns}+1:\ip{nt}} = H_{\ip{ns}+1:\ip{nt},j}(U_{ij})$ for all $j \in D$ and all $i \in \{\ip{ns}+1,\dots,\ip{nt}\}$. On the other hand,
\begin{multline*}
\psi_{C,A} \{ \B_n(s,t, \cdot) \} = \frac{1}{\sqrt{n}} \sum_{i=\ip{ns}+1}^{\ip{nt}}  \prod_{j \in A} (1 - U_{ij}) -  \sqrt{n} \lambda_n(s,t) \phi_A(C) \\ - \int_{[0,1]^d} \sum_{j \in A} \prod_{l \in A \setminus \{j\}} (1-v_l) \B_n(s,t,\vec v^{\{j\}}) \dd C(\vec v).
\end{multline*}
Next, let $\pi(\vec u) = \prod_{j \in A} (1-u_j)$, $\vec u \in \R^d$. Then, fix $\vec u \in [0,1]^d$, and, for any $x \in [0,1]$, let $\vec w_{\vec u}(x) = \vec u + x \{ \vec h_{\ip{ns}+1:\ip{nt}}(\vec u) - \vec u \}$ and let $g(x) = \pi\{ \vec w_{\vec u}(x) \}$, where $\vec h_{\ip{ns}+1:\ip{nt}}$ is defined in~\eqref{eq:hnkl}. The function $g$ is clearly continuously differentiable on $[0,1]$. By the mean value theorem, there exists $x_{\vec u,n,s,t}^* \in (0,1)$ such that $g(1) - g(0) = g'(x_{\vec u,n,s,t}^*)$, that is, such that
$$
\pi\{\vec h_{\ip{ns}+1:\ip{nt}}(\vec u)\} - \pi(\vec u) = \sum_{j \in A} \dot \pi_j [ \vec u + x_{\vec u,n,s,t}^* \{ \vec h_{\ip{ns}+1:\ip{nt}}(\vec u) - \vec u \} ] \{ H_{\ip{ns}+1:\ip{nt},j}(u_j) - u_j \}.
$$
It follows that 
\begin{multline*}
\S_{n,A}(s,t) - \psi_{C,A} \{ \B_n(s,t, \cdot)\} \\ = \frac{1}{\sqrt{n}} \sum_{i=\ip{ns}+1}^{\ip{nt}} \sum_{j \in A} \dot \pi_j [ \vec U_i + x_{\vec U_i,n,s,t}^* \{ \vec h_{\ip{ns}+1:\ip{nt}}(\vec U_i) - \vec U_i \} ] \{ H_{\ip{ns}+1:\ip{nt},j}(U_{ij}) - U_{ij} \} \\- \int_{[0,1]^d} \sum_{j \in A} \dot \pi_j(\vec v) \B_n(s,t,\vec v^{\{j\}}) \dd C(\vec v).
\end{multline*}
Notice that, by the triangle inequality and the fact that $\sup_{\vec u \in [0,1]^d} |\dot \pi_j(\vec u)| \leq 1$, $j \in D$,
$$
\sup_{(s,t) \in \Delta} |\S_{n,A}(s,t) - \psi_{C,A} \{ \B_n(s,t, \cdot)\} | \leq 2 |A| \sup_{(s,t,\vec u) \in \Delta \times [0,1]^d} |\B_n(s,t,\vec u)|.
$$
Next, fix $\eps, \eta > 0$. Using the previous inequality and the fact that $\B_n$ vanishes when $s=t$ and is asymptotically uniformly equicontinuous in probability as a consequence of Lemma~2 in \cite{Buc14}, there exists $\delta \in (0,1)$ such that, for all sufficiently large $n$,
\begin{multline*}
\Pr \left( \sup_{(s,t) \in \Delta \atop t-s < \delta} |\S_{n,A}(s,t) - \psi_{C,A} \{ \B_n(s,t, \cdot)\} | > \eps \right) \\ \leq \Pr \left( 2 |A| \sup_{(s,t,\vec u) \in \Delta \times [0,1]^d \atop t-s < \delta} |\B_n(s,t,\vec u)| > \eps \right) < \eta/2. 
\end{multline*}
To show~\eqref{eq:asymequivSnA}, it remains therefore to prove that, for all sufficiently large $n$,
$$
\Pr \left( \sup_{(s,t) \in \Delta \atop t-s \geq \delta} |\S_{n,A}(s,t) - \psi_{C,A} \{ \B_n(s,t, \cdot)\} | > \eps \right)  < \eta/2.
$$
To show the above, we shall now prove that $\sup_{(s,t) \in \Delta^\delta} |\S_{n,A}(s,t) - \psi_{C,A} \{ \B_n(s,t, \cdot)\} |$ converges in probability to zero, where $\Delta^\delta = \{(s,t) \in \Delta : t-s \geq \delta\}$. The latter supremum is smaller than $\sum_{j \in A} (I_{n,j} + II_{n,j})$, where
\begin{multline*}
I_{n,j} \leq \sup_{(s,t) \in \Delta^\delta} \Big| \frac{1}{\sqrt{n}} \sum_{i=\ip{ns}+1}^{\ip{nt}} \left(\dot \pi_j [ \vec U_i + x_{\vec U_i,n,s,t}^* \{ \vec h_{\ip{ns}+1:\ip{nt}}(\vec U_i) - \vec U_i \} ] - \dot \pi_j(\vec U_i) \right) \\ \times \{ H_{\ip{ns}+1:\ip{nt},j}(U_{ij}) - U_{ij} \} \Big|
\end{multline*}
and
\begin{multline*}
II_{n,j} \leq \sup_{(s,t) \in \Delta^\delta} \Big| \int_{[0,1]^d} \dot \pi_j(\vec v) \B_n(s,t,\vec v^{\{j\}}) \dd H_{\ip{ns}+1:\ip{nt}}(\vec v)- \int_{[0,1]^d} \dot \pi_j(\vec v) \B_n(s,t,\vec v^{\{j\}}) \dd C(\vec v) \Big|.
\end{multline*}
Next, notice that
\begin{multline}
\label{eq:wcHnsnt}
\sup_{(s,t,\vec u) \in \Delta^\delta \times [0,1]^d} |H_{\ip{ns}+1:\ip{nt}}(\vec u) - C(\vec u)| \\ \leq \sup_{(s,t,\vec u) \in \Delta^\delta \times [0,1]^d} |\B_n(s,t,\vec u)| \times n^{-1/2} \times \sup_{(s,t) \in \Delta^\delta} \{\lambda_n(s,t)\}^{-1} \p 0.
\end{multline}
Fix $j \in A$. Since the function $\dot \pi_j$ is continuous on $[0,1]^d$, by the continuous mapping theorem, $\sup_{(s,t,\vec u) \in \Delta^\delta \times [0,1]^d} | \dot \pi_j [ \vec u + x_{\vec u,n,s,t}^* \{ \vec h_{\ip{ns}+1:\ip{nt}}(\vec u) - \vec u \} ] - \dot \pi_j(\vec u) | \p 0$. Hence,
\begin{multline*}
I_{n,j} \leq \sup_{(s,t,\vec u) \in \Delta \times [0,1]^d} |\B_n(s,t,\vec u)| \\ \times \sup_{(s,t,\vec u) \in \Delta^\delta \in [0,1]^d} | \dot \pi_j [ \vec u + x_{\vec u,n,s,t}^* \{ \vec h_{\ip{ns}+1:\ip{nt}}(\vec u) - \vec u \} ] - \dot \pi_j(\vec u) | \p 0.
\end{multline*}
It thus remains to show that $II_{n,j} \p 0$.
The latter is mostly a consequence of Lemma~\ref{lem:prop1} below. First, notice that \eqref{eq:wcHnsnt} implies that $H_{\ip{ns}+1:\ip{nt}} \p C$ in $\ell^\infty(\Delta^\delta \times [0,1]^d;\R)$. Hence, $(\B_n,H_{\ip{ns}+1:\ip{nt}}) \leadsto (\B_C,C)$ in $\ell^\infty(\Delta^\delta \times [0,1]^d;\R)$. Next, combining the previous weak convergence with Lemma~3 in \citet{HolKojQue13} and the continuous mapping theorem, we obtain that the finite-dimensional distributions of $(\A_{n,j},\B_n)$ converge weakly to those of $(\A_{C,j},\B_C)$, where $\A_{n,j}$ and $\A_{C,j}$ are defined in Lemma~\ref{lem:prop1}. The fact that $(\A_{n,j},\B_n) \leadsto (\A_{C,j},\B_C)$ in $\{\ell^\infty(\Delta^\delta \times [0,1]^d;\R)\}^2$ then follows from Lemma~\ref{lem:prop1} below and the fact that marginal asymptotic tightness implies joint asymptotic tightness. The latter weak convergence combined with the continuous mapping theorem finally implies that $II_{n,j} \p 0$, which completes the proof.
\end{proof}

\begin{lem}
\label{lem:prop1}
For any $j \in D$ and $\delta \in (0,1)$, $\A_{n,j} \leadsto \A_{C,j}$ in $\ell^\infty(\Delta^\delta;\R)$, where 
\begin{align}
\label{eq:Anj}
\A_{n,j}(s,t) &= \int_{[0,1]^d} \dot \pi_j(\vec v) \B_n(s,t,\vec v^{\{j\}}) \dd H_{\ip{ns}+1:\ip{nt}}(\vec v), \\
\nonumber
\A_{C,j}(s,t) &= \int_{[0,1]^d} \dot \pi_j(\vec v) \B_C(s,t,\vec v^{\{j\}}) \dd C(\vec v).
\end{align}
\end{lem}

\begin{proof}
Fix $j \in D$ and $\delta \in (0,1)$. To prove the desired result, we shall show that conditions~(i) and~(ii) of Theorem 2.1 in \cite{Kos08} hold. First, recall that from~\eqref{eq:wcHnsnt}, $H_{\ip{ns}+1:\ip{nt}} \p C$ in $\ell^\infty(\Delta^\delta \times [0,1]^d;\R)$. Then, from the fact that $\B_n \leadsto \B_C$ in $\ell^\infty(\Delta \times [0,1]^d;\R)$, we obtain that, for any $(s_1,t_1),\dots,(s_k,t_k) \in \Delta^\delta$,
\begin{multline*}
\bigl( \B_n(s_1,t_1,\cdot),H_{\ip{ns_1}+1:\ip{nt_1}},\dots,\B_n(s_k,t_k,\cdot),H_{\ip{ns_k} + 1:\ip{nt_k}} \bigr) \\ 
\leadsto \bigl( \B_C(s_1,t_1,\cdot),C,\dots,\B_C(s_k,t_k,\cdot),C \bigr)
\end{multline*}
in $\{ \ell^\infty([0,1]^d;\R) \}^{2k}$. From Lemma~3 in \citet{HolKojQue13} and the continuous mapping theorem, the above implies that $\bigl( \A_{n,j}(s_1,t_1),\dots,\A_{n,j}(s_k,t_k) \bigr) \leadsto \bigl( \A_{C,j}(s_1,t_1),\dots,\A_{C,j}(s_k,t_k) \bigr)$ in $\R^k$. Hence, we have convergence of the finite-dimensional distributions, that is, condition~(i) of Theorem 2.1 in \cite{Kos08} holds.

It remains to prove condition~(ii) of Theorem 2.1 in \cite{Kos08}. Specifically, we shall now show that $\A_{n,j}$ is $\|\cdot\|_1$-asymptotically uniformly equicontinuous in probability, which will complete the proof since $\Delta^\delta$ is totally bounded by $\|\cdot\|_1$. By Problem 2.1.5 in \cite{vanWel96}, we need to show that, for any positive sequence $a_n \downarrow 0$,
\begin{equation}
\label{eq:asymequiAn}
\sup_{(s,t), (s',t') \in \Delta^\delta \atop |s-s'|+|t-t'| \leq a_n} | \A_{n,j}(s,t) - \A_{n,j}(s',t') | \p 0.
\end{equation}
We bound the supremum on the left of the previous display by $I_n + II_n$, where
\begin{multline*}
I_n = \sup_{(s,t), (s',t') \in \Delta^\delta \atop |s-s'|+|t-t'| \leq a_n} \left| \int_{[0,1]^d} \dot \pi_j(\vec v) \B_n(s,t,\vec v^{\{j\}}) \dd H_{\ip{ns}+1:\ip{nt}}(\vec v) \right. \\ \left. - \int_{[0,1]^d} \dot \pi_j(\vec v) \B_n(s',t',\vec v^{\{j\}}) \dd H_{\ip{ns}+1:\ip{nt}}(\vec v) \right|
\end{multline*}
and
\begin{multline*}
II_n =  \sup_{(s,t), (s',t') \in \Delta^\delta \atop |s-s'|+|t-t'| \leq a_n} \left| \int_{[0,1]^d} \dot \pi_j(\vec v) \B_n(s',t',\vec v^{\{j\}}) \dd H_{\ip{ns}+1:\ip{nt}}(\vec v) \right. \\ \left. - \int_{[0,1]^d} \dot \pi_j(\vec v) \B_n(s',t',\vec v^{\{j\}}) \dd H_{\ip{ns'}+1:\ip{nt'}}(\vec v) \right|.
\end{multline*}
Now, 
$$
I_n 
\leq \sup_{\vec u \in [0,1]^d} | \dot \pi_j(\vec u) | \times \sup_{(s,t), (s',t') \in \Delta^\delta, \vec u \in [0,1]^d \atop |s-s'|+|t-t'| \leq a_n} |  \B_n(s,t,\vec u) -  \B_n(s',t',\vec u) | \p 0,
$$
since $\B_n$ is asymptotically uniformly equicontinuous in probability as a consequence of Lemma~2 in \cite{Buc14}. Furthermore, $II_n$ is smaller than
{\small \begin{multline*}
\sup_{(s,t), (s',t') \in \Delta^\delta \atop |s-s'|+|t-t'| \leq a_n} \left| \frac{1}{\ip{nt}-\ip{ns}} \left\{ \sum_{i=\ip{ns}+1}^{\ip{nt}} \dot \pi_j(\vec U_i) \B_n(s',t',\vec U_i^{\{j\}}) - \sum_{i=\ip{ns'}+1}^{\ip{nt'}} \dot \pi_j(\vec U_i) \B_n(s',t',\vec U_i^{\{j\}}) \right\} \right| \\ + \sup_{(s,t), (s',t') \in \Delta^\delta \atop |s-s'|+|t-t'| \leq a_n} \left| \left( \frac{1}{\ip{nt}-\ip{ns}} - \frac{1}{\ip{nt'}-\ip{ns'}}  \right) \sum_{i=\ip{ns'}+1}^{\ip{nt'}} \dot \pi_j(\vec U_i) \B_n(s',t',\vec U_i^{\{j\}}) \right|,
\end{multline*}}
which is smaller than
\begin{multline*}
2 \times \sup_{(s,t), (s',t') \in \Delta^\delta \atop |s-s'|+|t-t'| \leq a_n} \frac{|\ip{nt} - \ip{nt'}| + |\ip{ns} - \ip{ns'}|}{\ip{nt}-\ip{ns}} \\ \times \sup_{\vec u \in [0,1]^d} | \dot \pi_j(\vec u) | \times \sup_{(s,t,\vec u) \in \Delta \times [0,1]^d} |  \B_n(s,t,\vec u) | \p 0.
\end{multline*}
Hence, $II_n \p 0$ and thus \eqref{eq:asymequiAn} holds, which completes the proof.
\end{proof}

\section{Proof of Corollary~\ref{cor:weak_TnA}}
\label{proof:cor:weak_TnA}

\begin{proof}
Starting from~\eqref{eq:TnAsH0}, using Proposition~\ref{prop:weak_SnA_sm}, the linearity of $\psi_{C,A}$ and~\eqref{eq:seqep}, we obtain that, for any $A \subseteq D$, $|A| \geq 1$, 
$$
\sup_{s \in [0,1]} | \T_{n,A}(s) - \psi_{C,A} \{ \B_n(0,s,\cdot) - \lambda(0,s) \B_n(0,1,\cdot) \} | = o_\Pr(1).
$$
Hence, $\T_n$ has the same weak limit as $s \mapsto \psi_C \{ \B_n(0,s,\cdot) - \lambda(0,s) \B_n(0,1,\cdot) \}$ and~\eqref{eq:wcTn} follows from the continuous mapping theorem.

The second to last claim is a consequence of the continuous mapping theorem. To prove the last claim, it suffices to show that the Gaussian process $\sigma_{C,f}^{-1} f \{ \T_C(\cdot)\}$ has the same covariance function as $\U$. For any, $s,t \in [0,1]$, we have
\begin{multline}
\label{eq:covexpr}
\cov[ \sigma_{C,f}^{-1} f\{\T_C(s)\},\sigma_{C,f}^{-1} f\{\T_C(t)\} ]  \\= \sigma_{C,f}^{-2} \Ex [ f \circ \psi_C \{ \B_C(0,s,\cdot)-s\B_C(0,1,\cdot) \} f \circ \psi_C \{ \B_C(0,t,\cdot)-t\B_C(0,1,\cdot) \}  ].
\end{multline}
By linearity of $f \circ \psi_C$ and Fubini's theorem, the expectation in the last display is equal to
$$
f \circ \psi_C \left\{ \vec u \mapsto f \circ \psi_C \left( \vec v \mapsto \Ex [ \{ \B_C(0,s,\vec u)-s\B_C(0,1,\vec u) \} \{ \B_C(0,t,\vec v)-t\B_C(0,1,\vec v) \} ] \right) \right\},
$$
that is, 
$$
(s \wedge t - st) f \circ \psi_C \left[ \vec u \mapsto f \circ \psi_C \left\{ \vec v \mapsto \kappa_C(\vec u,\vec v) \right\} \right] = (s \wedge t - st) \var [ f \circ \psi_C \{ \B_C(0,1,\cdot) \} ],
$$
where $\kappa_C$ is defined in~\eqref{eq:kappaC}. Combining the previous display with~\eqref{eq:covexpr}, we obtain that $\cov[ \sigma_{C,f}^{-1} f\{\T_C(s)\},\sigma_{C,f}^{-1} f\{\T_C(t)\} ] = (s \wedge t - st)$, which completes the proof.
\end{proof}

\section{Proofs of Propositions~\ref{prop:multTn} and~\ref{prop:multbarTn}}
\label{proof:prop:multTn}

\begin{proof}[\bf Proof of Proposition~\ref{prop:multTn}]
We only show the first claim as the subsequent claims then mostly follow from the continuous mapping theorem. Also, we only provide the proof under~(ii) in the statement of Proposition~\ref{prop:multBn}, the proof being simpler under~(i). Fix $A \subseteq D$, $|A| \geq 1$.  For any $(s,t) \in \Delta$, let $\S_{n,A}^{(m)}(s,t) = \psi_{C,A}\{\check \B_n^{(m)}(s,t,\cdot)\}$. Using the linearity of the map $\psi_{C,A}$ defined in~\eqref{eq:psiCA}, Proposition~\ref{prop:multBn} and the continuous mapping theorem, we obtain that 
$$
\left(\S_{n,A}, \S_{n,A}^{(1)}, \dots, \S_{n,A}^{(M)} \right) \leadsto \left( \S_{C,A}, \S_{C,A}^{(1)}, \dots, \S_{C,A}^{(M)} \right)
$$
in $\{\ell^\infty(\Delta;\R)\}^{M+1}$. The first claim is thus proved if we show that, for any $m \in \{1,\dots,M\}$, $\sup_{(s,t) \in \Delta} | \check \S_{n,A}^{(m)}(s,t) -  \S_{n,A}^{(m)}(s,t)|$ is $o_\Pr(1)$. Fix $m \in \{1,\dots,M\}$ and notice that the latter supremum is smaller than $2 |A| \sup_{(s,t,\vec u) \in \Delta \times [0,1]^d} |\check \B_n^{(m)}(s,t,\vec u)|$. We can therefore proceed analogously to the proof of Proposition~\ref{prop:weak_SnA_sm}. Fix $\eps, \eta > 0$. Using the previous inequality as well as the fact that $\check \B_n^{(m)}$ is zero when $s=t$ and is asymptotically uniformly equicontinuous in probability as a consequence of Lemma~A.3 in \cite{BucKoj14}, there exists $\delta \in (0,1)$ such that, for all sufficiently large $n$,
$$
\Pr \left( \sup_{(s,t) \in \Delta \atop t-s < \delta} | \check \S_{n,A}^{(m)}(s,t) -  \S_{n,A}^{(m)}(s,t)| > \eps \right) < \eta/2. 
$$
It remains therefore to prove that $\sup_{(s,t) \in \Delta^\delta} |\check \S_{n,A}^{(m)}(s,t) -  \S_{n,A}^{(m)}(s,t)| \p 0$, where $\Delta^\delta = \{(s,t) \in \Delta : t-s \geq \delta\}$. The latter supremum is smaller than
$$
\sum_{j \in A} \sup_{(s,t) \in \Delta^\delta} \Big| \int_{[0,1]^d} \dot \pi_j(\vec v) \check \B_n^{(m)}(s,t,\vec v^{\{j\}}) \dd C_{\ip{ns}+1:\ip{nt}}(\vec v)- \int_{[0,1]^d} \dot \pi_j(\vec v) \check \B_n^{(m)}(s,t,\vec v^{\{j\}}) \dd C(\vec v) \Big|, 
$$ 
where $\dot \pi_j$ is the $j$th first order partial derivative of the function $\pi(\vec u) = \prod_{j \in A} (1-u_j)$, $\vec u \in \R^d$, introduced in the proof of Proposition~\ref{prop:weak_SnA_sm}. Fix $j \in A$. The $j$th summand in the previous display is smaller than $I_n$ + $II_n$, where
\begin{align*}
I_n &= \sup_{(s,t) \in \Delta^\delta} \Big| \int_{[0,1]^d} \dot \pi_j(\vec v) \check \B_n^{(m)}(s,t,\vec v^{\{j\}}) \dd C_{\ip{ns}+1:\ip{nt}}(\vec v) - \check \A_{n,j}^{(m)}(s,t) \Big|, \\
II_n &= \sup_{(s,t) \in \Delta^\delta} \Big| \check \A_{n,j}^{(m)}(s,t) -  \int_{[0,1]^d} \dot \pi_j(\vec v) \check \B_n^{(m)}(s,t,\vec v^{\{j\}}) \dd C(\vec v) \Big|,
\end{align*}
and $\check \A_{n,j}^{(m)}$ is defined analogously to the process $\A_{n,j}$ in~\eqref{eq:Anj} with $\B_n$ replaced by $\check \B_n^{(m)}$. In addition, it can be verified that Lemma~\ref{lem:prop1} remains true if $\B_n$ and $\B_C$ are replaced by $\check \B_n^{(m)}$ and $\B_C^{(m)}$, respectively, in its statement. It follows that we can proceed as at the end of proof of Proposition~\ref{prop:weak_SnA_sm} to show that $II_n$ above converges to zero in probability.

To show that $I_n \p 0$, we use the fact that $I_n \leq I_n' + I_n''$, where
\begin{align*}
I_n' =& \sup_{(s,t) \in \Delta^\delta} \Big| \frac{1}{\ip{nt} - \ip{ns}} \sum_{i=\ip{ns}+1}^{\ip{nt}} \left[ \dot \pi_j\{ \vec h_{\ip{ns}+1:\ip{nt}}(\vec U_i) \} - \dot \pi_j (\vec U_i) \right] \\ &\times \check \B_n^{(m)} \{ s,t, \vec h_{\ip{ns}+1:\ip{nt}}(\vec U_i)^{\{j\}} \}  \Big|, \\
I_n'' =& \sup_{(s,t) \in \Delta^\delta} \Big| \frac{1}{\ip{nt} - \ip{ns}} \sum_{i=\ip{ns}+1}^{\ip{nt}} \dot \pi_j (\vec U_i) \left[ \check \B_n^{(m)} \{ s,t, \vec h_{\ip{ns}+1:\ip{nt}}(\vec U_i)^{\{j\}} \} - \check \B_n^{(m)} ( s,t, \vec U_i^{\{j\}} ) \right] \Big|. 
\end{align*}
For $I_n'$, we have that 
$$
I_n' \leq \sup_{(s,t,\vec u) \in \Delta \times [0,1]^d} \left|\check \B_n^{(m)}(s,t,\vec u) \right| \times \sup_{(s,t,\vec u) \in \Delta^\delta \times [0,1]^d} \left| \dot \pi_j\{ \vec h_{\ip{ns}+1:\ip{nt}}(\vec u) \} - \dot \pi_j (\vec u) \right| \p 0
$$
as a consequence of the weak convergence of $\check \B_n^{(m)}$,~\eqref{eq:wcHnsnt}, and the continuous mapping theorem. For $I_n''$, using the fact that $\sup_{\vec u \in [0,1]^d} |\dot \pi_j(\vec u)| \leq 1$, we obtain that
$$
I_n'' \leq \sup_{(s,t,\vec u) \in \Delta^\delta \times [0,1]^d} \left| \check \B_n^{(m)} \{ s,t,\vec h_{\ip{ns}+1:\ip{nt}}(\vec u)^{\{j\}} \} - \check \B_n^{(m)} ( s,t, \vec u^{\{j\}} ) \right| \p 0.
$$
The latter convergence is a consequence of the asymptotic equicontinuity in probability of $\check \B_n^{(m)}$ and the fact that $\sup_{(s,t,u) \in \Delta^\delta \times [0,1]} |H_{\ip{ns}+1:\ip{nt},j}(u) - u| \p 0$ \citep[see e.g.\ the treatment of the term (B.9) in][for a detailed proof of a similar convergence]{BucKojRohSeg14}.
\end{proof}

\begin{proof}[\bf Proof of Proposition~\ref{prop:multbarTn}]
We only provide the proof under~(ii) in the statement of Proposition~\ref{prop:multBn}, the proof being simpler under~(i). From Proposition~\ref{prop:multTn}, to prove the desired result it suffices to show that, for any $A \subseteq D$, $|A| \geq 1$,
$$
\sup_{(s,t) \in \Delta} | \tilde{\S}_{n,b_n,A}^{(m)}(s,t) - \check{\S}_{n,A}^{(m)}(s,t) | \p 0.
$$
Fix $A \subseteq D$, $|A| \geq 1$. From~\eqref{eq:checkSnA} and~\eqref{eq:tildeSnA} and the triangle inequality, the latter will hold if, for any $j \in A$, 
$$
\sup_{(s,t,u) \in \Delta \times [0,1]} | \tilde{\B}_{n,b_n,j}^{(m)}(s,t,u) - \check{\B}_n^{(m)}(s,t,\vec u_j) | \p 0.
$$
The previous supremum can actually be restricted to $u \in (0,1)$ as both processes are zero if $u \in \{0,1\}$.

Let $K > 0$ be a constant and let us first suppose that, for any  $n \geq 1$ and $i \in \{1,\dots,n\}$, $\xi_{i,n}^{(m)} \geq -K$. Also, fix $j \in A$. The supremum on the right of the previous display is then smaller than $I_n + II_n$, where
\begin{align*}
I_n &= \sup_{(s,t,u) \in \Delta \times (0,1)}  \frac{1}{\sqrt{n}} \sum_{i = \ip{ns}+1}^{\ip{nt}} ( \xi_{i,n}^{(m)} + K ) \left| \LL_{b_n}( \hat U_{ij}^{\ip{ns}+1:\ip{nt}},u ) - \1( \hat U_{ij}^{\ip{ns}+1:\ip{nt}} \leq u ) \right|, \\
II_n &= \sup_{(s,t,u) \in \Delta \times (0,1)} \frac{K + \bar  \xi_{\ip{ns}+1:\ip{nt}}^{(m)} }{\sqrt{n}} \sum_{i = \ip{ns}+1}^{\ip{nt}}  \left| \LL_{b_n}( \hat U_{ij}^{\ip{ns}+1:\ip{nt}},u ) - \1( \hat U_{ij}^{\ip{ns}+1:\ip{nt}} \leq u ) \right|. 
\end{align*}
Next, some thought reveals that, for any $(u,v) \in [0,1] \times (0,1)$, 
\begin{align}
\label{eq:ineqLL} 
\left| \LL_{b_n}(u,v) - \1(u \leq v) \right| &\leq  \1(u_- \leq v)  - \1(u_+ \leq v) \\\nonumber
&=  \1(u - b_n \leq v)  - \1(u + b_n \leq v) \\
\nonumber
&=  \1(u \leq v_+)  - \1(u \leq v_-). 
\end{align}
Then, we write $I_n \leq I_{n,1} + I_{n,2}$, where
\begin{align*}
I_{n,1} &= \sup_{(s,t,u) \in \Delta \times [0,1]}  \left|  \frac{1}{\sqrt{n}} \sum_{i = \ip{ns}+1}^{\ip{nt}} ( \xi_{i,n}^{(m)} - \bar  \xi_{\ip{ns}+1:\ip{nt}}^{(m)}) \1(u_- < \hat U_{ij}^{\ip{ns}+1:\ip{nt}} \leq u_+ )  \right|, \\
I_{n,2} &= \sup_{(s,t,u) \in \Delta \times [0,1]} \frac{K + \bar  \xi_{\ip{ns}+1:\ip{nt}}^{(m)} }{\sqrt{n}} \sum_{i = \ip{ns}+1}^{\ip{nt}}  \1(u_- < \hat U_{ij}^{\ip{ns}+1:\ip{nt}} \leq u_+ ). 
\end{align*}
For $I_{n,1}$, we have
$$
I_{n,1} \leq \sup_{(s,t,\vec u, \vec v) \in \Delta \times [0,1]^{2d} \atop \|\vec u - \vec v\|_1 \leq 2b_n} \left| \check{\B}_n^{(m)}(s,t,\vec u) - \check{\B}_n^{(m)}(s,t,\vec v) \right| \p 0 
$$
from the asymptotic uniform equicontinuity in probability of $\check{\B}_n^{(m)}$. Before dealing with $I_{n,2}$, let us first show that 
\begin{equation}
\label{eq:In3}
I_{n,3} = \sup_{(s,t,u) \in \Delta \times [0,1]} \frac{1}{\sqrt{n}} \sum_{i = \ip{ns}+1}^{\ip{nt}}  \1(u_- < \hat U_{ij}^{\ip{ns}+1:\ip{nt}} \leq u_+ ) \p 0.
\end{equation}
From the proof of Proposition~3.3 of \cite{BucKojRohSeg14}, we have that
$$
\sup_{(s,t,u) \in \Delta \times [0,1]} \left| \frac{1}{\sqrt{n}} \sum_{i = \ip{ns}+1}^{\ip{nt}}  \left[ \1\{ U_{ij} \leq H_{\ip{ns}+1:\ip{nt},j}^{-1}(u) \} - \1(\hat U_{ij}^{\ip{ns}+1:\ip{nt}} \leq u) \right] \right| \p 0.
$$
Consequently, to prove that $I_{n,3} \p 0$, it suffices to show that
$$
\sup_{(s,t,u) \in \Delta \times [0,1]} \left| \frac{1}{\sqrt{n}} \sum_{i = \ip{ns}+1}^{\ip{nt}}  \left[ \1\{ U_{ij} \leq H_{\ip{ns}+1:\ip{nt},j}^{-1}(u_+) \} -  \1\{ U_{ij} \leq H_{\ip{ns}+1:\ip{nt},j}^{-1}(u_-) \} \right] \right|\p 0.
$$
The supremum on the left of the previous display is smaller than $J_{n,1} + J_{n,2} + 
J_{n,3}$, where
\begin{align*}
J_{n,1} &= \sup_{(s,t,u) \in \Delta \times [0,1]} \left| \B_n\{ s,t,1,H_{\ip{ns}+1:\ip{nt},j}^{-1}(u_+),1 \} - \B_n\{ s,t,1,H_{\ip{ns}+1:\ip{nt},j}^{-1}(u_-),1\}  \right|, \\ 
J_{n,2} &= \sup_{(s,t,u) \in \Delta \times [0,1]}  \sqrt{n} \lambda_n(s,t) \left| H_{\ip{ns}+1:\ip{nt},j}^{-1}(u_+) - u_+ - H_{\ip{ns}+1:\ip{nt},j}^{-1}(u_-) + u_- \right|, \\ 
J_{n,3} &= \sup_{(s,t,u) \in \Delta \times [0,1]}  \sqrt{n} \lambda_n(s,t) \left| u_+ - u_- \right|,
\end{align*}
with some abuse of notation for $J_{n,1}$. We immediately have $J_{n,3} \leq 2 \sqrt{n} b_n \to 0$. The fact $J_{n,2} \p 0$ follows from the asymptotic uniform equicontinuity in probability of the process $(s,t,u) \mapsto \sqrt{n} \lambda_n(s,t) \{ H_{\ip{ns}+1:\ip{nt},j}^{-1}(u) - u \}$, itself following from its weak convergence to $(s,t,u) \mapsto - \B_C(s,t,\vec u_j)$ in $\ell^\infty(\Delta \times [0,1];\R)$. The latter is a consequence of the weak convergence of $\B_n$ to $\B_C$ in $\ell^\infty(\Delta \times [0,1]^d;\R)$, Lemma~B.2 of \cite{BucKoj14} and the extended continuous mapping theorem \cite[Theorem 1.11.1]{vanWel96}. The fact that $J_{n,2} \p 0$ implies that, for any $\delta \in (0,1)$, 
$$
\sup_{(s,t,u) \in \Delta \times [0,1] \atop t-s \geq \delta} \left| H_{\ip{ns}+1:\ip{nt},j}^{-1}(u_+) - H_{\ip{ns}+1:\ip{nt},j}^{-1}(u_-) \right| \p 0.
$$
Combined with the asymptotic uniform equicontinuity in probability of $\B_n$, the latter can be used to prove that $J_{n,1} \p 0$ \citep[see][page 24, term (B.9), for a similar proof]{BucKojRohSeg14}. Hence, $I_{n,3} \p 0$.

Now, $I_{n,2} \leq K \times I_{n,3} + I_{n,4}$, where
$$
I_{n,4} = \sup_{(s,t,u) \in \Delta \times [0,1]} \frac{\bar  \xi_{\ip{ns}+1:\ip{nt}}^{(m)} }{\sqrt{n}} \sum_{i = \ip{ns}+1}^{\ip{nt}}  \1(u_- < \hat U_{ij}^{\ip{ns}+1:\ip{nt}} \leq u_+ ). 
$$
Hence, to show that $I_{n,2} \p 0$, it remains to prove that $I_{n,4} \p 0$. The latter can be shown by proceeding as for the term~(B.8) in \cite{BucKojRohSeg14}.

We therefore have that $I_n \p 0$. The fact that $II_n \p 0$, follows from the fact that $II_n \leq I_{n,2} \p 0$. This completes the proof under the condition $\xi_{i,n}^{(m)} \geq -K$. To show that this condition is not necessary, we use the arguments employed at the end of the proof of Proposition 4.3 of \cite{BucKojRohSeg14}. 
\end{proof}

\section{Proofs of Propositions~\ref{prop:convsigma} and~\ref{prop:convsigmatilde}}
\label{proof:prop:convsigma}

\begin{lem}
\label{lem:wcHn}
Assume that $\vec U_1,\dots,\vec U_n$ is drawn from a strictly stationary sequence $(\vec U_i)_{i \in \Z}$ whose strong mixing coefficients satisfy $\alpha_r = O(r^{-a})$, $a > 6$. Then, for any $A \subseteq D$, $|A| \geq 1$ and $j \in A$, $\Hb_{n,A,j} \leadsto \Hb_{A,j}$ in $\ell^\infty([0,1];\R)$, where, for any $t \in [0,1]$, $\Hb_{n,A,j}(t) = n^{-1/2} \sum_{i=1}^n \left[ Y_{i,A,j}(t) - \Ex \{ Y_{1,A,j}(t) \} \right]$, $Y_{i,A,j}(t) = \prod_{l \in A \setminus \{j\}} (1-U_{il}) \1(t \leq U_{ij})$, and $\Hb_{A,j}$ is a tight process. 
\end{lem}

\begin{proof}
Fix $A \subseteq D$, $|A| \geq 1$ and $j \in A$. To simplify the notation, we write $\Hb_n$ instead of $\Hb_{n,A,j}$ and $Y_i$ instead of $Y_{i,A,j}$ as we continue. To prove the desired result, we mostly adapt the arguments used in the proof of Proposition 2.11 of \cite{DehPhi02}. From Theorem 2.1 in \cite{Kos08}, two conditions are needed to obtain the desired weak convergence. The first condition (which is the weak convergence of the finite-dimensional distributions) is a consequence of Theorem 3.23 of \cite{DehPhi02} as $a > 6$ and $Y_i(t) \in [0,1]$ for all $t \in [0,1]$. To prove the second condition, we shall show that $\Hb_n$ is asymptotically $|\cdot|$-equicontinuous in probability. To do so, we shall first prove that, for any $\eps, \delta > 0$, there exists a grid $0 = t_0 < t_1 < \dots < t_k = 1$ such that, for all $n$ sufficiently large,
\begin{equation}
\label{eq:goal}
\Pr \left\{  \max_{1 \leq i \leq k} \sup_{t \in [t_{i-1},t_i]} | \Hb_n(t) - \Hb_n(t_{i-1}) | \geq \eps \right\} \leq \delta.
\end{equation}
We first note that there exists constants $c \geq 1$ and $\epsilon \in (0,1)$ such that $\alpha_r \leq c r^{-6-\epsilon}$. Then, using the fact that, for $t,t' \in [0,1]$,
$$
\Ex[ \{ Y_1(t) - Y_1(t') \}^2 ] \leq \Ex[ | Y_1(t) - Y_1(t') | ] \leq \Ex\{ \1(t \wedge t' \leq U_{ij} \leq t \vee t') \} = |t - t'|,
$$
we apply Lemma~3.22 of \cite{DehPhi02} with $\xi_i = Y_i(t) - Y_i(t')$ to obtain that
$$
\Ex[ \{ \Hb_n(t) - \Hb_n(t') \}^4 ] \leq 10^4 \frac{c}{\epsilon} \left( |t - t'|^\eta + n^{-1} |t - t'|^{\eta/2}  \right) = \lambda \left( |t - t'|^\eta + n^{-1} |t - t'|^{\eta/2}  \right),
$$
where $\eta = 1 + \epsilon/10 > 1$ and $\lambda = 10^4 c/\epsilon$. It follows that, for any $t,t' \in [0,1]$ such that $|t - t'| \geq n^{-2/\eta}$, 
\begin{equation}
\label{eq:4moment}
\Ex[ \{ \Hb_n(t) - \Hb_n(t') \}^4 ] \leq 2 \lambda |t - t'|^\eta.
\end{equation}
Next, consider a grid $0 = t_0 < t_1 < \dots < t_k = 1$ to be specified later. Furthermore, it can be verified that the function $G:t \mapsto \Ex\{Y_1(t)\}$ is continuous and strictly decreasing on $[0,1]$. Then, fix $i \in \{1,\dots,k\}$, let $\tau = \eps n^{-1/2}/4$, let $m = m_i = \ip{ \{ G(t_{i-1}) - G(t_i) \} / \tau }$ and define a subgrid $t_{i-1} = s_0 < s_1 < \dots < s_m = t_i$ such that $G(s_j) = G(s_0) - j \tau$ for $j \in \{1,\dots,m-1\}$. Notice that this ensures that, for any $j \in \{1,\dots,m\}$, $\tau \leq G(s_{j-1}) - G(s_j) \leq 2 \tau$. Now, fix $j \in \{1,\dots,m\}$. Using the fact that the function $t \mapsto n^{-1} \sum_{i=1}^n Y_i(t)$ is also decreasing, it can be verified that, for any $t \in [s_{j-1}, s_j]$,
$$
\Hb_n(t) - \Hb_n(t_{i-1}) \leq |\Hb_n(s_{j-1}) - \Hb_n(t_{i-1})| + \eps/2
$$
and
$$
- \eps/2 -  |\Hb_n(s_j) - \Hb_n(t_{i-1})| \leq \Hb_n(t) - \Hb_n(t_{i-1}).
$$
The above inequalities imply that, for any $t \in [t_{i-1},t_i] = \bigcup_{j=1}^m [s_{j-1},s_j]$,
$$
- \eps/2  + \min_{1 \leq j \leq m} \{ -  |\Hb_n(s_j) - \Hb_n(t_{i-1})| \} \leq \Hb_n(t) - \Hb_n(t_{i-1}) \leq \max_{2 \leq j \leq m} |\Hb_n(s_{j-1}) - \Hb_n(t_{i-1})| + \eps/2,
$$
and thus that
$$
\sup_{t \in [t_{i-1},t_i]} | \Hb_n(t) - \Hb_n(t_{i-1}) | \leq  \max_{1 \leq j \leq m} |\Hb_n(s_j) - \Hb_n(t_{i-1})| + \eps/2.
$$
Hence,
\begin{equation} 
\label{eq:probineq}
\Pr \left\{  \sup_{t \in [t_{i-1},t_i]} | \Hb_n(t) - \Hb_n(t_{i-1}) | \geq \eps \right\} \leq \Pr \left\{  \max_{1 \leq j \leq m} |\Hb_n(s_j) - \Hb_n(t_{i-1})| \geq \eps/2 \right\}.
\end{equation}
Now, let $\zeta_l = \Hb_n(s_l) - \Hb_n(s_{l-1})$, $l \in \{1,\dots,m\}$ with $\zeta_0 = 0$, and let $S_j = \sum_{l=0}^j \zeta_l$, $j \in \{0,\dots,m\}$. From~\eqref{eq:4moment}, we then have that, for any $0 \leq j < j' \leq m$ and $n$ sufficiently large, 
\begin{multline*}
\Ex\{ (S_{j'} - S_{j})^4 \} = \Ex\left\{ \left(\sum_{l=j+1}^{j'} \zeta_l \right)^4 \right\} = \Ex\left[ \left \{ \Hb_n(s_{j'}) -\Hb_n(s_{j})  \right\}^4 \right] \\ \leq 2 \lambda (s_{j'} - s_j)^\eta = 2 \lambda \left\{ \sum_{j < l \leq j'}
 (s_l - s_{l-1}) \right\}^\eta.
\end{multline*}
Indeed, by construction of the subgrid, for any $0 \leq j < j' \leq m$, $ n^{-1/2} \eps/4 \leq G(s_j) - G(s_{j'}) \leq s_{j'} - s_j$, and $n^{-1/2} \eps/4$ can be made larger than $n^{-2/\eta}$  by taking $n$ sufficiently large since $2/\eta > 1/2$. The assumption of Theorem~2.12 of \cite{Bil68} being satisfied \citep[see also Lemma~2.10 in][]{DehPhi02}, we obtain that there exists a constant $K \geq 0$ such that, for any $\nu \geq 0$,
$$
\Pr \left( \max_{1 \leq j \leq m} |S_j| \geq \nu \right) \leq \nu^{-4} K (s_m - s_0)^\eta = \nu^{-4} K (t_i - t_{i-1})^\eta.
$$ 
Applying the previous inequality to the right-hand side of~\eqref{eq:probineq}, we obtain that  
$$
\Pr \left\{  \sup_{t \in [t_{i-1},t_i]} | \Hb_n(t) - \Hb_n(t_{i-1}) | \geq \eps \right\} \leq \eps^{-4} 2^4 K (t_i - t_{i-1})^\eta.
$$
It follows that 
\begin{align*}
\Pr &\left\{  \max_{1 \leq i \leq k} \sup_{t \in [t_{i-1},t_i]} | \Hb_n(t) - \Hb_n(t_{i-1}) | \geq \eps \right\} \leq \eps^{-4} 2^4 K \sum_{i=1}^k (t_i - t_{i-1})^\eta \\
&\leq \eps^{-4} 2^4 K \times \max_{1 \leq i \leq k} (t_i - t_{i-1})^{\eta-1} \times \sum_{i=1}^k (t_i - t_{i-1}).
\end{align*}
By choosing the initial grid such that $\max_{1 \leq i \leq k} (t_i - t_{i-1}) \leq \{ \delta \eps^4 2^{-4} K^{-1} \}^{1/(\eta-1)}$, we obtain~\eqref{eq:goal}.

It remains to verify that $\Hb_n$ is asymptotically $|\cdot|$-equicontinuous in probability. By Problem 2.1.5 in \cite{vanWel96}, this amounts to showing that for any positive sequence $a_n \downarrow 0$ and any $\eps,\delta > 0$, 
\begin{equation}
\label{eq:goal2}
\Pr \left\{ \sup_{s,t \in [0,1] \atop |t - s| \leq a_n} | \Hb_n(s) - \Hb_n(t) | > 3 \eps \right\} \leq \delta
\end{equation}
for $n$ sufficiently large. Fix $\eps,\delta > 0$ and $a_n \downarrow 0$, and choose a grid $0 = t_0 < \dots < t_k = 1$ such that~\eqref{eq:goal} holds for all $n$ sufficiently large. Furthermore, let $\mu = \min_{1 < i < k} (t_i - t_{i-1})$. Then, from \citet[Theorem 7.4]{Bil99}, we have that, for all $n$ sufficiently large such that $a_n \leq \mu$, 
$$
\sup_{s,t \in [0,1] \atop |t - s| \leq a_n} | \Hb_n(s) - \Hb_n(t) | \leq 3 \max_{1 \leq i \leq k} \sup_{t \in [t_{i-1},t_i]} | \Hb_n(t) - \Hb_n(t_{i-1}) |.
$$
Finally,~\eqref{eq:goal2} follows for all $n$ sufficiently large by combining the previous inequality with~\eqref{eq:goal}.
\end{proof}

\begin{proof}[\bf Proof of Proposition~\ref{prop:convsigma}]
We shall only prove the result under~(ii), the proof being simpler under~(i). Recall $\sigma_{n,C,f}^2$ defined in~\eqref{eq:sigmaBn}. From~\eqref{eq:MSE}, we immediately have that $\sigma_{n,C,f}^2 \p \sigma_{C,f}^2$. 
It remains to show that $\check \sigma_{n,C_{1:n},f}^2 - \sigma_{n,C,f}^2 \p 0$. 

Recall $\vec h_{1:n}$ defined in~\eqref{eq:hnkl} and that $\pobs{U}_i^{1:n} = \vec h_{1:n}(\vec U_i)$ for all $i \in \{1,\dots,n\}$. Then, starting from~\eqref{eq:checksigma} and~\eqref{eq:sigma_nCf}, it can be verified that
\begin{multline}
\label{eq:ineqsigma}
| \check \sigma_{n,C_{1:n},f}^2 - \sigma_{n,C,f}^2  | \leq \left\{ \frac{1}{n} \sum_{i,j=1}^n \varphi \left(\frac{i-j}{\ell_n}\right) \right\}\\ 
\times  \left[ \sup_{\vec u \in [0,1]^d} |f \{ \I_C(\vec u) - \psi_{C}(C) \} |  + \sup_{\vec u \in [0,1]^d} \left| f [ \I_{C_{1:n}} \{ \vec h_{1:n}(\vec u) \} - \psi_{C_{1:n}}(C_{1:n}) ] \right| \right] \\ \times \sup_{\vec u \in [0,1]^d} \left| f [ \I_{C_{1:n}} \{ \vec h_{1:n}(\vec u) \} - \I_C(\vec u) - \psi_{C_{1:n}}(C_{1:n}) + \psi_{C}(C)   ] \right| .
\end{multline}
Some algebra shows that the second term on the right of the previous inequality is smaller than
$$
\sup_{\vec u \in [0,1]^d} | f \circ \I_C(\vec u) | + | f \circ \psi_C(C) |  + 2 \sup_{\vec u \in [0,1]^d} | f \circ \I_{C_{1:n}}(\vec u) |. 
$$
From~\eqref{eq:ICA} and~\eqref{eq:psiCA}, we have that, for any $A \subseteq D$, $|A| \geq 1$,  $\sup_{\vec u \in [0,1]^d} |\I_{C,A}(\vec u)| \leq 1$, $\sup_{\vec u \in [0,1]^d} |\I_{C_{1:n},A}(\vec u)| \leq 1$ and $|\psi_{C,A}(C)| \leq 1$. Hence, by~\eqref{eq:psiC},~\eqref{eq:IC} and linearity of $f$, we have that the second term (between square brackets) on the right of inequality~\eqref{eq:ineqsigma} is bounded by $4 \sup_{\vec x \in [-1,1]^{2^d-1}}|f(\vec x)| < \infty$. Concerning the first term on the right of~\eqref{eq:ineqsigma}, we have
$$
\frac{1}{n} \sum_{i,j=1}^n \varphi \left(\frac{i-j}{\ell_n}\right) = \frac{1}{n} \sum_{k=-\ell_n}^{\ell_n} (n - |k|) \varphi \left(\frac{k}{\ell_n}\right)  \leq 2 \ell_n + 1 = O(n^{1/2-\eps}).
$$
We will now show that the last supremum on the right of~\eqref{eq:ineqsigma} is $O_\Pr(n^{-1/2})$, which will complete the proof. By the triangle inequality,
\begin{multline*}
\sup_{\vec u \in [0,1]^d} \left| f [ \I_{C_{1:n}} \{ \vec h_{1:n}(\vec u) \} - \I_C(\vec u) - \psi_{C_{1:n}}(C_{1:n}) + \psi_{C}(C)   ] \right| \\ \leq \sup_{\vec u \in [0,1]^d} \left| f [ \I_{C_{1:n}} \{ \vec h_{1:n}(\vec u) \} - \I_C(\vec u) ] \right| +   \left| f \{ \psi_{C_{1:n}}(C_{1:n}) - \psi_{C}(C) \} \right|.
\end{multline*}
By linearity of $f$, from~\eqref{eq:ICA} and~\eqref{eq:IC}, to show that the first term on the right on the previous inequality is $O_\Pr(n^{-1/2})$, it suffices to show that, for any $A \subseteq D$, $|A| \geq 1$, 
\begin{equation}
\label{eq:show1}
\sup_{\vec u \in [0,1]^d} \left| \I_{C_{1:n},A} \{ \vec h_{1:n}(\vec u) \} - \I_{C,A}(\vec u) \right| = O_\Pr(n^{-1/2}).
\end{equation} 
Similarly, for the second term on the right, it suffices to show that, for any $A \subseteq D$, $|A| \geq 1$, $\left| \psi_{C_{1:n},A}(C_{1:n}) - \psi_{C,A}(C) \right| = O_\Pr(n^{-1/2})$. Now, from Fubini's theorem, $\psi_{C,A}(C) = \psi_{C,A}[\Ex\{\1(\vec U_1 \leq \cdot) \}] = \Ex\{\I_{C,A}(\vec U_1)\}$. Hence, $\left| \psi_{C_{1:n},A}(C_{1:n}) - \psi_{C,A}(C) \right|$ is smaller than
\begin{multline*}
\left| \frac{1}{n} \sum_{i=1}^n \left\{ \I_{C_{1:n},A}(\pobs{U}_i^{1:n}) - \I_{C,A}(\vec U_i) \right\} \right| + \left| \frac{1}{n} \sum_{i=1}^n \left[ \I_{C,A}(\vec U_i) - \Ex\{\I_{C,A}(\vec U_1)\}  \right] \right| \\
\leq \sup_{\vec u \in [0,1]^d} \left| \I_{C_{1:n},A} \{ \vec h_{1:n}(\vec u) \} - \I_{C,A}(\vec u) \right| + \left| \frac{1}{n} \sum_{i=1}^n \left[ \I_{C,A}(\vec U_i) - \Ex\{\I_{C,A}(\vec U_1)\}  \right] \right|.
\end{multline*}
The proof is therefore complete if we show~\eqref{eq:show1} and that the second term on the right of the previous inequality is $O_\Pr(n^{-1/2})$. The latter is a consequence of the weak convergence of $n^{-1/2} \sum_{i=1}^n \left[ \I_{C,A}(\vec U_i) - \Ex\{\I_{C,A}(\vec U_1)\} \right]$ which follows from Theorem 3.23 of \cite{DehPhi02} as a consequence of the fact that $\sup_{\vec u \in [0,1]^d} |\I_{C,A}(\vec u)| \leq 1$ and the assumption on the mixing rate.

It remains to prove~\eqref{eq:show1}. The latter will follow by the triangle inequality if we show that, for any $A \subseteq D$, $|A| \geq 1$,
\begin{align}
\label{eq:one}
&\sup_{\vec u \in [0,1]^d} | \I_{C,A} \{ \vec h_{1:n}(\vec u) \} - \I_{C,A} (\vec u) | = O_\Pr(n^{-1/2}), \\
\label{eq:two}
&\sup_{\vec u \in [0,1]^d} | \I_{H_{1:n},A} (\vec u) - \I_{C,A} (\vec u) | = O_\Pr(n^{-1/2}), \\
\label{eq:three}
&\sup_{\vec u \in [0,1]^d} | \I_{C_{1:n},A} (\vec u) -  \I_{H_{1:n},A} (\vec u) | = O_\Pr(n^{-1/2}).
\end{align}

Fix $A \subseteq D$, $|A| \geq 1$. 

{\em Proof of~\eqref{eq:one}.} We have
\begin{multline*}
\sup_{\vec u \in [0,1]^d} | \I_{C,A} \{ \vec h_{1:n}(\vec u) \} - \I_{C,A} (\vec u) |  \leq \sup_{\vec u \in [0,1]^d} \left| \prod_{l \in A} \{1 - H_{1:n,l}(u_l) \} - \prod_{l \in A} (1 - u_l) \right| \\ + \sum_{j \in A} \sup_{u \in [0,1]} \left| \int_{[0,1]^d} \prod_{l \in A \setminus \{j\}} (1 - v_l) \left[ \1 \{ H_{1:n,j}(u) \leq v_j \} - \1 ( u \leq v_j ) \right] \dd C(\vec v) \right|.
\end{multline*} 
By an application of the mean value theorem similar to that performed in the proof of Proposition~\ref{prop:weak_SnA_sm}, it is easy to verify that the first supremum is $O_\Pr(n^{-1/2})$ since, for any $j \in D$, $\sup_{u \in [0,1]} |H_{1:n,j}(u) - u| = O_\Pr(n^{-1/2})$ as a consequence of the weak convergence of $\B_n$ defined in~\eqref{eq:seqep}. The second term is smaller than
\begin{multline*}
\sum_{j \in A} \sup_{u \in [0,1]}  \int_{[0,1]} \left| \1 \{ H_{1:n,j}(u) \leq v \} - \1 ( u \leq v ) \right| \dd v \\ \leq \sum_{j \in A} \sup_{u \in [0,1]}  \int_{[0,1]} \1 \{ u \wedge H_{1:n,j}(u) \leq v \leq  u \vee H_{1:n,j}(u) \} \dd v \\= \sum_{j \in A} \sup_{u \in [0,1]} | H_{1:n,j}(u) -  u  | = O_\Pr(n^{-1/2}).
\end{multline*}

{\em Proof of~\eqref{eq:two}:} From~\eqref{eq:ICA} and the triangle inequality, it suffices to show that, for any $j \in A$, 
$$
\sup_{u \in [0,1]} \left| \frac{1}{n} \sum_{i=1}^n \prod_{l \in A \setminus \{j\}} (1-U_{il}) \1(u \leq U_{ij}) - \int_{[0,1]^d} \prod_{l \in A \setminus \{j\}} (1-v_l) \1(u \leq v_j) \dd C(\vec v) \right| = O_\Pr(n^{-1/2}).
$$
The latter is an immediate consequence of the weak convergence result stated in Lemma~\ref{lem:wcHn} and the continuous mapping theorem.

{\em Proof of~\eqref{eq:three}:} The supremum on the left of~\eqref{eq:three} is smaller than $I_n + II_n + III_n$, where 
\begin{align}
\nonumber
I_n &= \sup_{\vec u \in [0,1]^d} | \I_{C_{1:n},A}(\vec u) -  \I_{H_{1:n},A} \{ \vec h_{1:n}^{-1}(\vec u) \} |, \\ 
\label{eq:IIn}
II_n &= \sup_{\vec u \in [0,1]^d} | \I_{H_{1:n},A} \{ \vec h_{1:n}^{-1}(\vec u) \} - \I_{C,A}\{ \vec h_{1:n}^{-1}(\vec u) \} -   \I_{H_{1:n},A} (\vec u) + \I_{C,A}(\vec u)  |, \\   
\label{eq:IIIn}
III_n &= \sup_{\vec u \in [0,1]^d} | \I_{C,A}\{ \vec h_{1:n}^{-1}(\vec u) \} -  \I_{C,A} (\vec u) |, 
\end{align}
with $\vec h_{1:n}^{-1}$ is defined in~\eqref{eq:hnklinv}. The term $I_n$ is smaller
\begin{multline*}
\sup_{\vec u \in (0,1]^d} \left| \prod_{l \in A} (1 - u_l) - \prod_{l \in A} \{ 1 - H_{1:n,l}^{-1}(u_l) \}  \right| \\ + \sup_{\vec u \in [0,1]^d} \left| \frac{1}{n} \sum_{i=1}^n \sum_{j \in A} \prod_{l \in A \setminus \{j\}} \{ 1-H_{1:n,l}(U_{il}) \} \1 \{ u_j \leq H_{1:n,j}(U_{ij}) \} \right. \\ \left. - \frac{1}{n} \sum_{i=1}^n \sum_{j \in A} \prod_{l \in A \setminus \{j\}} (1-U_{il}) \1 \{ H_{1:n,j}^{-1}(u_j) \leq U_{ij} \} \right|.
\end{multline*}
Since, for any $j \in D$, $\sup_{u \in [0,1]} |H_{1:n,j}^{-1}(u) - u| = \sup_{u \in [0,1]} |H_{1:n,j}(u) - u|$ (for instance, by symmetry arguments on the graphs of $H_{1:n,j}$ and $H_{1:n,j}^{-1}$), and by an application of the mean value theorem as above, we obtain that the first supremum is $O_\Pr(n^{-1/2})$. Using the fact that, for all $u \in [0,1]$, $u \leq H_{1:n,j}(U_{ij})$ is equivalent to $H_{1:n,j}^{-1}(u) \leq U_{ij}$, it can be verified that the second supremum is smaller than
\begin{align*}
\sum_{j \in A} \sup_{u \in [0,1]} \left| \frac{1}{n} \sum_{i=1}^n \left[ \prod_{l \in A \setminus \{j\}} \{ 1-H_{1:n,l}(U_{il}) \} - \prod_{l \in A \setminus \{j\}} (1-U_{il}) \right] \1\{u \leq H_{1:n,j}(U_{ij})\} \right| \\ \leq \sum_{j \in A} \sup_{\vec u \in [0,1]^d} \left| \prod_{l \in A \setminus \{j\}} \{ 1-H_{1:n,l}(u_l) \} - \prod_{l \in A \setminus \{j\}} (1-u_l) \right| = O_\Pr(n^{-1/2}),
\end{align*}
where the last equality follows again by an application of the mean value theorem as above. Hence, $I_n = O_\Pr(n^{-1/2})$. For $II_n$ defined in~\eqref{eq:IIn}, we have
$$
II_n \leq n^{-1/2} \sum_{j \in A} \sup_{u \in [0,1]} \left| \Hb_{n,A,j}\{H_{1:n,j}^{-1}(u)\} - \Hb_{n,A,j}(u) \right| = o_\Pr(n^{-1/2}),
$$
where $\Hb_{n,A,j}$ is defined in Lemma~\ref{lem:wcHn}. The last equality is a consequence of the asymptotic equicontinuity in probability of $\Hb_{n,A,j}$ and the fact that $\sup_{u \in [0,1]} |H_{1:n,j}^{-1}(u) - u| = \sup_{u \in [0,1]} |H_{1:n,j}(u) - u| \as 0$. The latter convergence follows from the almost sure invariance principle established in \cite{BerPhi77} and \cite{Yos79}. It implies a functional law of the iterated logarithm for $u \mapsto H_{1:n,j}(u) - u$ as soon as $a > 3$, which in turn implies the Glivenko--Cantelli lemma under strong mixing.

It remains to show that $III_n$ defined in~\eqref{eq:IIIn} is $O_\Pr(n^{-1/2})$. The proof of the latter is similar to that of~\eqref{eq:one}. 
\end{proof}

\begin{proof}[\bf Proof of Proposition~\ref{prop:convsigmatilde}]
We only show the result under~(ii), the proof being simpler under~(i). To prove the desired result, we shall show that $\tilde \sigma_{n,b_n,C_{1:n},f}^2 - \check \sigma_{n,C_{1:n},f}^2 \p 0$. Proceeding as in the proof of Proposition~\ref{prop:convsigma} for~\eqref{eq:ineqsigma}, it can be verified that to prove the above, it suffices to show that, for any $A \subseteq D$, $|A| \geq 1$,
$$
\sup_{\vec u \in [0,1]^d} \left| \I_{b_n,C_{1:n},A} (\vec u) - \I_{C_{1:n},A}(\vec u) \right| = O_\Pr(n^{-1/2}).
$$
Fix $A \subseteq D$, $|A| \geq 1$. From~\eqref{eq:ICA} and~\eqref{eq:IbnCA}, we have that the supremum on the right of the previous display is smaller than $\sum_{j \in A} I_{n,j}$, where
$$
I_{n,j} = \sup_{u \in [0,1]} \int_{[0,1]^d} \left| \LL_{b_n} (u,v_j) - \1(u \leq v_j) \right| \dd C_{1:n}(\vec v).
$$
Fix $j \in A$. From~\eqref{eq:ineqLL}, we have that $I_{n,j} \leq n^{-1/2} J_{n,j}$, where 
\begin{align*}
J_{n,j} &= \sup_{u \in [0,1]} \frac{1}{\sqrt{n}} \sum_{i=1}^n \{ \1(u_- \leq \hat U_{ij}^{1:n} ) - \1(u_+ \leq \hat U_{ij}^{1:n} ) \} \\
&= \sup_{u \in [0,1]} \frac{1}{\sqrt{n}} \sum_{i=1}^n \{ \1(\hat U_{ij}^{1:n} < u_+) - \1(\hat U_{ij}^{1:n} < u_-) \} \\
&\leq \sup_{u \in [0,1]} \frac{1}{\sqrt{n}} \sum_{i=1}^n \{ \1(\hat U_{ij}^{1:n} \leq u_+) - \1(\hat U_{ij}^{1:n} \leq u_-) \} + \sup_{u \in [0,1]} \frac{1}{\sqrt{n}} \sum_{i=1}^n \1(\hat U_{ij}^{1:n} = u). 
\end{align*}
Proceeding as for~\eqref{eq:In3}, we obtain that the first supremum on the right of the previous display converges in probability to zero. The second supremum is smaller than
$$
\sup_{u \in [0,1]} \frac{1}{\sqrt{n}} \sum_{i=1}^n \{ \1(\hat U_{ij}^{1:n} \leq u) - \1(\hat U_{ij}^{1:n} \leq u - 1/n) \}
$$
and can be dealt with along the same lines. Hence, $J_{n,j} \p 0$, which implies that $I_{n,j} = o(n^{-1/2})$ and completes the proof.
\end{proof}

\end{document}